\documentclass[11pt,letterpaper]{amsart}

\usepackage{url}
\usepackage{fullpage}
\usepackage{amsmath}

\usepackage{algpseudocode,algorithm,algorithmicx}

\algrenewcommand\algorithmicrequire{\textbf{Description:}}
\algrenewcommand\algorithmicensure{\textbf{Postcondition:}}
\algnewcommand{\IIf}[1]{\State\algorithmicif\ #1\ \algorithmicthen}
\algnewcommand{\EndIIf}{\unskip\ \algorithmicend\ \algorithmicif}

\usepackage{amssymb}
\usepackage{amsthm}
\usepackage{color}
\usepackage{array}
\usepackage{xy}
\usepackage{setspace}
\usepackage{multicol}
\usepackage{hyperref}
\usepackage{cite}
\usepackage[margin=1in,letterpaper]{geometry}
\usepackage{accents}

\usepackage{todonotes}

\newcommand{\poly}{\operatorname{poly}}
\newcommand{\E}{\mathbb{E}}

\renewcommand{\epsilon}{\varepsilon}
\newcommand{\polylog}{\operatorname{polylog}}
\newcommand{\defn}[1]{\textbf{\emph{#1}}} \renewcommand{\paragraph}[1]{\vspace{.5 cm} \noindent \textbf{#1.} }

\newtheoremstyle{slanted}
{3pt}
{3pt}
{\slshape}
{}
{\bfseries}
{.}
{.5em}
{}
\theoremstyle{slanted}
\newtheorem{theorem}{Theorem}[section]
\newtheorem{definition}[theorem]{Definition}
\newtheorem{lemma}[theorem]{Lemma}

\newtheorem{claim}[theorem]{Claim}

\theoremstyle{remark}
\newtheorem{rem}[theorem]{Remark}

\usepackage{hyperref}

\begin{document}

\title[How Asymmetry Helps Buffer Management: Achieving Optimal Tail Size in Cup Games]{How Asymmetry Helps Buffer Management: Achieving Optimal Tail Size in Cup Games}
\author{William Kuszmaul}
\thanks{MIT CSAIL. \textit{kuszmaul@mit.edu}. Supported by a Fannie \& John Hertz Foundation Fellowship, a NSF GRFP Fellowship, and NSF grant CCF 1533644.}

\maketitle

\thispagestyle{empty}
\begin{abstract}

The cup game on $n$ cups is a multi-step game with two players, a filler and an emptier. At each step, the filler distributes $1$ unit of water among the cups, and then the emptier selects a single cup to remove (up to) $1$ unit of water from.

There are several objective functions that the emptier might wish to minimize. One of the strongest guarantees would be to minimize \emph{tail size}, which is defined to be the number of cups with fill $2$ or greater. A simple lower-bound construction shows that the optimal tail size for deterministic emptying algorithms is $\Theta(n)$, however.

We present a simple randomized emptying algorithm that achieves tail size $\tilde{O}(\log n)$ with high probability in $n$ for $\poly n$ steps. Moreover, we show that this is tight up to doubly logarithmic factors. We also extend our results to the \emph{multi-processor cup game}, achieving tail size $\tilde{O}(\log n + p)$ on $p$ processors with high probability in $n$. We show that the dependence on $p$ is near optimal for any emptying algorithm that achieves polynomial-bounded backlog.

A natural question is whether our results can be extended to give \emph{unending guarantees}, which apply to arbitrarily long games. We give a lower bound construction showing that no monotone memoryless emptying algorithm can achieve an unending guarantee on either tail size or the related objective function of backlog. On the other hand, we show that even a very small (i.e., $1 / \poly n$) amount of resource augmentation is sufficient to overcome this barrier.
\end{abstract}

\newpage

\setcounter{page}{1}
\section{Introduction}

At the start of the \defn{cup game on $n$ cups}, there are $n$ empty
cups. In each step of the game, a \defn{filler} distributes $1$ unit
of water among the cups, and then an \defn{emptier} removes (up to) $1$
unit of water from a single cup of its choice. The emptier aims to
minimize some measure of ``behind-ness'' for cups in the system (e.g.,
the height of the fullest cup, or the number of cups above a certain
height). If the emptier's algorithm is randomized, then the filler is
an \defn{oblivious adversary}, meaning it cannot adapt to the
behavior of the emptier.

The cup game naturally arises in the study of processor
scheduling, modeling a problem in which $n$ tasks each receive work at
varying rates, and a scheduler must pick one task to schedule at each
time step
\cite{BaruahCoPl96,GkasieniecKl17,BaruahGe95,LitmanMo11,LitmanMo05,MoirRa99,BarNi02,GuanYi12,Liu69,
  LiuLa73, AdlerBeFr03, LitmanMo09, DietzRa91}. The cup game has also found
numerous applications to network-switch buffer
management~\cite{Goldwasser10,AzarLi06,RosenblumGoTa04,Gail93},
quality of service
guarantees~\cite{BaruahCoPl96,AdlerBeFr03,LitmanMo09}, and data
structure deamortization
\cite{AmirFaId95,DietzRa91,DietzSl87,AmirFr14,Mortensen03,GoodrichPa13,FischerGa15,Kopelowitz12,BenderDaFa20}.

\paragraph{Bounds on backlog}
Much of the work on cup games has focused on bounding the
\defn{backlog} of the system, which is defined to be the amount of
water in the fullest cup.

Research on bounding backlog has spanned five decades
\cite{BaruahCoPl96,GkasieniecKl17,BaruahGe95,LitmanMo11,LitmanMo05,MoirRa99,BarNi02,GuanYi12,Liu69,
  LiuLa73,DietzRa91, BenderFaKu19, Kuszmaul20, AdlerBeFr03, DietzSl87,
  LitmanMo09}. Much of the early work focused on the \defn{fixed-rate}
version of the game, in which the filler places a fixed amount of
water $f_j$ into each cup $j$ on every step
\cite{BaruahCoPl96,GkasieniecKl17,BaruahGe95,LitmanMo11,LitmanMo05,MoirRa99,BarNi02,GuanYi12,Liu69,
  LiuLa73}; in this case constant backlog is achievable  \cite{Liu69, BaruahCoPl96}. For the full
version of the game, without fixed rates, constant backlog is not
possible. In this case, the optimal deterministic emptying algorithm
is known to be the greedy emptying algorithm, which always empties
from the fullest cup, and which achieves backlog $O(\log n)$
\cite{AdlerBeFr03, DietzSl87}. If the emptier is permitted to use a
randomized algorithm, then it can do much better, achieving an
asymptotically optimal backlog of $O(\log \log n)$ for $\poly n$ steps
with high probability \cite{DietzRa91, BenderFaKu19, Kuszmaul20}.

\paragraph{A strong guarantee: small tail size}
The \defn{tail size} of a cup game at each step is the number of cups
containing at least some constant $c$ amount of water. For the
guarantees in this paper, $c$ will be taken to be $2$.

A guarantee of small tail size is particularly appealing for
scheduling applications, where cups represent tasks and water
represents work that arrives to the tasks over time. Whereas a bound
of $b$ on backlog guarantees that the furthest behind worker is only
behind by at most $b$, it says nothing about the \emph{number} of
workers that are behind by $b$. In contrast, a small bound on tail
size ensures that \emph{almost no workers} are behind by more than
$O(1)$.

The main result in this paper is a randomized emptying algorithm that
achieves tail size \\ $O(\log n \log \log n)$. The algorithm also
simultaneously optimizes backlog, keeping the maximum height at
$O(\log \log n)$. As a result the \emph{total amount of water} above
height $2$ in the system is $\tilde{O}(\log n)$ with high
probability. In constrast, the best possible deterministic emptying
algorithm allows for up to $n^{1- \epsilon}$ cups to \emph{all} have
fills $\Omega(\log n)$ at once (see the lower-bound constructions
discussed in \cite{DietzRa91} and \cite{BenderFeKr15}).

The problems of optimizing tail size and backlog are closely related
to the problem of optimizing the \defn{$c$-shifted $\ell_p$ norm} of the cup game. Formally, the $c$-shifted $\ell_p$ norm is given by
$$\left(\sum_{i=1}^n \max(f_i - c, 0)^p\right)^{1/p},$$
where $f_i$ is the fill of cup $i$. \footnote{When $p \neq \infty$, in order for the 
$c$-shifted $\ell_p$ norm to be interesting, it is necessary to require that
$c \geq \Omega(1)$, since trivial filling strategies can achieve 
fill $\Omega(1)$ in $\Theta(n)$ cups deterministically.}
The problem of bounding backlog corresponds to the problem of optimizing the
$\ell_\infty$ norm of the cup game, and the problem of bounding tail size corresponds
to bounding the $2$-shifted $\ell_1$ norm. By optimizing both metrics 
simultaneously our algorithm has the desirable property that 
it also achieves a bound of $\tilde{O}(\log n)$ for the $2$-shifted $\ell_p$ norm 
for \emph{any} $p \in O(1)$, which is optimal up to doubly logarithmic factors.

\paragraph{Past work on tail size using beyond-worst-case analysis}
To approach the combinatorial difficulty of analyzing cup games, past
works have often turned to various forms of beyond worst-case analysis
(e.g., smoothed analysis \cite{BenderFaKu19}, semi-clairvoyance \cite{LitmanMo09}, resource augmentation \cite{BenderFaKu19, LitmanMo09, DietzRa91, DietzSl87}). 
The most successful of these approaches has arguably been resource augmentation. 
In the cup game with $\epsilon$ resource
augmentation, the filler is permitted to distribute at most $1 - \epsilon$ water into cups
at each step (rather than $1$ full unit), giving the emptier a small advantage.
Resource augmentation can
also be studied in a more extreme form by allowing the emptier to
fully empty a cup on each step, rather than simply removing a single
unit \cite{DietzRa91, DietzSl87}. Resource augmentation significantly simplifies the
task of analyzing cup games --- for example, there was nearly a 30 year
gap between the first randomized bounds on backlog with resource
augmentation \cite{DietzRa91} versus the first bounds without resource augmentation \cite{Kuszmaul20}.

Currently, the best known guarantees with resource augmentation are
achieved by the algorithm of \cite{BenderFaKu19} which, using
$\epsilon = 1 / \polylog n$, limits backlog to $O(\log \log n)$ and
tail size to $\polylog n$ (with high probability).

The algorithm,
which is called the \defn{smoothed greedy algorithm}, begins by
randomly perturbing the starting state of the system, and then
following a variant of the deterministic greedy emptying
algorithm.
Roughly speaking, the random perturbations at the beginning of the
game allow for the number of cups $X_t$ containing more than $1$ unit
of water after each step $t$ to be modeled by a biased random walk
$X_1, X_2, \ldots$, where $\Pr[X_i = X_{i - 1} + 1] = 1/2 - \epsilon$
and $\Pr[X_i = \max(X_{i - 1} - 1, 0)] = 1/2 + \epsilon$. This is
where the resource augmentation plays a critical role, since it
introduces a bias to the random walk which pushes the walk near $0$ at
all times. In contrast, without resource augmentation, the random walk
is unbiased.

Subsequent work \cite{Kuszmaul20} showed that the algorithm's
guarantees on backlog continue to hold without resource augmentation
(at least, for a polynomial number of steps). More generally, the
analysis bounds the number of cups at a given height $\alpha$ by
roughly $n^{1 / 2^{\Omega(\alpha)}}$, which in turn implies an
arbitrarily small polynomial tail size. Whether or not a subpolynomial
tail size can be achieved without resource augmentation has remained
open.

\paragraph{This paper: the asymmetric smoothed greedy algorithm}
We show that resource augmentation is not needed to bound tail
size. We present a randomized emptying algorithm (called the \defn{asymmetric smoothed greedy algorithm}) that achieves tail
size $\tilde{O}(\log n)$ with high probability in $n$ after each of
the first $\poly n$ steps of a cup game. We prove that the algorithm
is nearly optimal, in that any emptying algorithm must allow for a
tail size of $\tilde{\Omega}(\log n)$ with probability at least
$\frac{1}{\poly n}$.

The analysis of the algorithm takes an indirect approach to bounding tail size. Rather than examining tail size directly, we instead
prove that the use of randomness by the algorithm makes the state of
the cups at each step ``difficult'' for the filler to accurately
predict. We call this the \defn{unpredictability guarantee}. We then
show that any greedy-like algorithm that satisfies the
unpredictability guarantee is guaranteed to have tail size
$\tilde{O}(\log n)$ and backlog $O(\log \log n)$.

Algorithmically, the key to achieving the unpredictability guarantee
is to add a small amount of asymmetry to the smoothed greedy emptying algorithm. When choosing between cups that have 2 or more units of water, the emptier follows the standard smoothed greedy algorithm; but when choosing between cups whose fills are between 1 and 2, the emptier ignores the specifics of how much water is in each cup, and instead chooses between the cups based on random priorities that are assigned to the cups at the beginning of the game.

Intuitively, the asymmetric treatment of cups ensures that
there is a large collection (of size roughly $n / 2$) of randomly
selected cups that are ``almost always'' empty. The fact that the emptier doesn't know which cups these are then implies the unpredictability guarantee. 
Proving this intuition remains highly nontrivial, however, and requires several new combinatorial ideas.

\paragraph{Multi-processor guarantees}
The cup game captures a scheduling problem in which a single processor
must pick one of $n$ tasks to make progress on in each time step. The
multi-processor version of this scheduling problem is captured by the
\defn{$p$-processor cup game} \cite{BenderFaKu19, Kuszmaul20,
  LitmanMo09,AdlerBeFr03,BaruahCoPl96,Liu69}. In each step of the
$p$-processor cup game, the filler distributes $p$ units of water
among cups, and the emptier removes $1$ unit of water from (up to) $p$
different cups. Because the emptier can remove at most $1$ unit of
water from each cup at each step, an analogous constraint is also
placed on the filler, requiring that it places at most $1$ unit of
water into each cup at each step.

A key feature of the $p$-processor cup game is that the emptier is
required to remove water from $p$ \emph{distinct} cups in each step,
even if the vast majority of water is contained in fewer than $p$
cups. 

Until recently, establishing any nontrivial bounds on backlog in the
multi-processor cup game remained an open problem, even with the help
of resource augmentation. Recent work by Bender et
al. \cite{BenderFaKu19} (using resource augmentation) and then by
Kuszmaul \cite{Kuszmaul20} (without resource augmentation) established
bounds on backlog closely matching those for the single-processor
game.

By extending our techniques to the multi-processor setting, we
construct a randomized emptying algorithm achieves tail size
$\tilde{O}(\log n + p)$ with high probability in $n$ after each of the
first $\poly n$ steps of a $p$-processor cup game. Moreover, we show
that the dependence on $p$ is near optimal for any backlog-bounded
algorithm (i.e., any algorithm that achieves backlog $\poly n$ or
smaller).

\paragraph{Lower bounds against unending guarantees}
In the presence of resource augmentation $\epsilon = 1 / \polylog n$,
the smoothed greedy emptying algorithm is known to provide an
\emph{unending guarantee} \cite{BenderFaKu19}, meaning that the
high-probability bounds on backlog and tail size continue to hold even
for \emph{arbitrarily} large steps $t$.

A natural question is whether unending guarantees can also be achieved
\emph{without} the use of resource augmentation. It was previously
shown that, when $p \ge 2$, the smoothed greedy algorithm does not
offer unending guarantees \cite{Kuszmaul20}. Analyzing the
single-processor game has remained an open question, however.

We give a lower bound construction showing that neither the smoothed
greedy algorithm nor the asymmetric smoothed greedy algorithm
offer unending guarantees for the single-processor cup game without
the use of resource augmentation. Even though resource augmentation
$\epsilon > 0$ is needed for the algorithms to achieve unending
guarantees, we show that the amount of resource augmentation required
is very small. Namely, $\epsilon = 1 / 2^{\polylog n}$ is both
sufficient and necessary for the asymmetric smoothed greedy
algorithm to offer unending guarantees on both tail size and backlog.

We generalize our lower-bound construction to work against \emph{any
  emptying algorithm} that is both monotone and memoryless, including
emptying algorithms that are equipped with a clock. We show that no
such emptying algorithm can offer an unending guarantee of $o(\log n)$
backlog in the single-processor cup game, and that any unending
guarantee of $\polylog n$ tail size must come of the cost of
polynomial backlog.

We call the filling strategy in our lower bound construction the
\defn{fuzzing algorithm}. The fuzzing algorithm takes a very simple
approach: it randomly places water into a pool of cups, and shrinks
that pool of cups very slowly over time. The fact that gradually
shrinking random noise represents a worst-case workload for cup games
suggests that real-world applications of cup games (e.g., processor scheduling,
network-switch buffer management, etc.) may be at risk of experiencing
``aging'' over time, with the performance of the system degrading due
to the impossibility of strong unending guarantees.

\paragraph{Related work on other variants of cup games}
Extensive work has also been performed on other variants of the cup
game. Bar-Noy et al.~\cite{Bar-NoyFrLa02} studied the backlog for a
variant of the single-processor cup game in which the filler can place
arbitrarily large integer amounts of water into cups at each
step. Rather than directly bounding the backlog, which would be
impossible, they show that the greedy emptying algorithm achieves
competitive ratio $O(\log n)$, and that this is optimal for both
deterministic and randomized online emptying algorithms. Subsequent
work has also considered weaker adversaries~\cite{FleischerKo04,
  DamaschkeZh05}.

Several papers have also explored variants of cup games in which cups are
connected by edges in a graph, and in which the emptier is constrained
by the structure of the
graph~\cite{BodlaenderHuKu12,BenderFeKr15,BodlaenderHuWo11,ChrobakCsIm01}. This
setting models multi-processor scheduling with conflicts between
tasks~\cite{BodlaenderHuWo11,ChrobakCsIm01} and some problems in
sensor radio networks~\cite{BenderFeKr15}.

Another recent line of work is that by Kuszmaul and Westover \cite{variable}, which considers a variant of the $p$-processor cup game in which the filler is permitted to change the value of $p$ over time. Remarkably, the optimal backlog in this game is significantly worse than in the standard game, and is $\Theta(n)$ for an (adaptive) filler.

Cup games have also been used to model memory-access heuristics in
databases ~\cite{BenderCrCo18}. Here, the emptier is allowed to
completely empty a cup at each step, but the water from that cup is
then ``recycled'' among the cups according to some probability
distribution. The emptier's goal is achieve a large recycling rate,
which is the average amount of water recycled in each step.

Closely related to the study of cup games is the problem of \defn{load balancing}, in which one must assign balls to bins in order to minimize the number of balls in the fullest bin. In the classic load balancing problem, $n$ balls arrive over time, and each ball comes with a selection of $d$ random bins (out of $n$ bins) in which it can potentially be placed. 
The load balancing algorithm gets to select which of the $d$ bins to place the ball in, and can, for example, always select the bin with the fewest balls. But what should the algorithm do when choosing between bins that have the same number of balls? In this case, V\"ocking famously showed that the algorthm should always break ties in the same direction \cite{vocking2003asymmetry}, and that this actually results in an asymptotically better bound on load than if one breaks ties arbitrarily. Interestingly, one can think of the asymmetry used in V\"ocking's algorithm for load balancing as being analogous to the asymmetry used in our algorithm for the cup game: in both cases, the algorithm always breaks ties in the same random direction, although in our result, the way that one should define a ``tie'' is slightly nonobvious. In the case of V\"ocking's result, the asymmetry is known to be necessary in order to get an optimal algorithm \cite{vocking2003asymmetry}; it remains an open question whether the same is true for the problem of bounding tail size in cup games.

\paragraph{Related work on the roles of backlog and tail size in data structures}
Bounds on backlog have been used extensively in data-structure
deamortization
\cite{AmirFaId95,DietzRa91,DietzSl87,AmirFr14,Mortensen03,GoodrichPa13,FischerGa15,Kopelowitz12},
where the scheduling decisions by the emptier are used to decide how a
data structure should distribute its work.

Until recently, the applications focused primarily on in-memory data
structures, since external-memory data structures often cannot afford
the cost of a buffer overflowing by an $\omega(1)$ factor. Recent work
shows how to use bounds in tail-size in order to solve this problem,
and presents a new technique for applying cup games to external-memory
data structures \cite{BenderDaFa20}. A key insight is that if a cup
game has small tail size, then the water in ``overflowed cups'' (i.e.,
cups with fill more than $O(1)$) can be stored in a small in-memory
cache. The result is that every cup consumes exactly $\Theta(1)$
blocks in external memory, meaning that each cup can be read/modified
by the data structure in $O(1)$ I/Os. This insight was recently
applied to external-memory dictionaries in order to eliminate flushing
cascades in write optimized data structures \cite{BenderDaFa20}.

\paragraph{Outline}
The paper is structured as follows. Section \ref{sec:alg} describes a
new randomized algorithm that achieves small tail size without
resource augmentation. Section \ref{sec:technical} gives a technical
overview of the algorithm's analysis and of the other results in this
paper. Section \ref{sec:analysis} then presents the full analysis of
the algorithm and Section \ref{sec:lowerbounds} presents (nearly)
matching lower bounds. Finally, Section \ref{sec:unending} gives lower
bounds against unending guarantees and analyzes the amount of resource
augmentation needed for such guarantees.


\paragraph{Conventions}
Although in principle an arbitrary constraint height $c$ can be used to
determine which cups contribute to the tail size, all of the
algorithms in this paper work with $c = 2$. Thus, throughout the rest
of the paper, we define the \defn{tail size} to be the number of cups
with height $2$ or greater.

As a convention, we say that an event occurs with \defn{high
  probability in $n$}, if the event occurs with probability at least
$1 - \frac{1}{n^c}$ for an arbitrarily large constant $c$ of our
choice. The constant $c$ is allowed to affect other constants in the
statement. For example, an algorithm that achieves tail size $c \log
n$ with probability $\frac{1}{n^c}$ is said to achieve tail size
$O(\log n)$ with high probability in $n$.

\section{The Asymmetric Smoothed Greedy Algorithm}\label{sec:alg}

Past work on randomized emptying algorithms has focused on analyzing
the \defn{smoothed greedy algorithm} \cite{BenderFaKu19,
  Kuszmaul20}. The algorithm begins by randomly perturbing the
starting state of the system: the emptier  places a
random offset $r_j$ of water into each cup $j$, where the $r_j$'s are
selected independently and uniformly from $[0, 1)$. The emptier then
follows a greedy emptying algorithm, removing water from the fullest
cup at each step. If the fullest cup contains fill less than $1$,
however, then the emptier skips its turn. This ensures that the
\defn{fractional} amount of water in each cup $j$ (i.e., the amount of
water modulo $1$) is permanently randomized by the initial offset
$r_j$. The randomization of the fractional amounts of water in each
cup has been critical to past randomized analyses \cite{BenderFaKu19,
  Kuszmaul20}, and continues to play an important (although perhaps less
central) role in this paper.

This paper introduces a new variant of the smoothed greedy
algorithm that we call the \defn{asymmetric smoothed greedy
  algorithm}. The algorithm assigns a random priorities
$p_j \in [0 ,1)$ to each cup $j$ (at the beginning of the game) and
uses these to ``break ties'' when cups contain relatively small
amounts of water. Interestingly, by always breaking these ties in the
same direction, we change the dynamics of the game in a way that
allows for new analysis techniques. We describe the algorithm in
detail below.

\paragraph{Algorithm description} At the beginning of the game, the emptier selects random offsets $r_j
\in [0, 1)$ independently and uniformly at random for each cup
  $j$. Prior to the game beginning, $r_j$ units of water are placed in
  each cup $j$. This water is for ``bookkeeping'' purposes only, and
  need not physically exist. During initialization, the emptier also
  assigns a random \defn{priority} $p_j \in [0, 1)$ independently and
    uniformly at random to each cup $j$.

After each step $t$, the emptier selects (up to) $p$ different cups to
remove $1$ unit of water from as follows. If there are $p$ or more
cups containing $2$ or more units of water, then the emptier selects
the $p$ fullest such cups. Otherwise, the emptier selects all of the
cups containing $2$ or more units of water, and then resorts to cups
containing fill in $[1, 2)$, choosing between these cups based on
  their priorities $p_j$ (i.e., choosing cups with larger priorities
  over those with smaller priorities). The emptier never removes water
  from any cup containing less than $1$ unit of water.

  \paragraph{Threshold crossings and a threshold queue}
  When discussing the algorithm, several additional definitions and
  conventions are useful.  We say that \defn{threshold $(j, i)$ is
    crossed} if cup $j$ contains at least $i$ units of water for
  positive integer $i$. When $i = 1$, the threshold $(j, i)$ is called
  a \defn{light threshold}, and otherwise $(j, i)$ is called a
  \defn{heavy threshold}. One interpretation of the emptying algorithm
  is that there is a queue $Q$ of thresholds $(j, i)$ that are
  currently crossed. Whenever the filler places water into cups, this
  may add thresholds $(j, i)$ to the queue. And whenever the emptier
  removes water from some cup $j$, this removes some threshold
  $(j, i)$ from the queue. When selecting thresholds to remove from
  the queue, the emptier prioritizes heavy thresholds over light
  ones. Within the heavy thresholds, the emptier prioritizes based on
  cup height, and within the light thresholds the emptier prioritizes
  based on cup priorities $p_j$.

  As a convention, we say that a cup $j$ is \defn{queued} if $(j, 1)$
  is in $Q$ (or, equivalently, if $(j, i)$ is in the queue for any
  $i$). The emptier is said to \defn{dequeue} cup $j$ whenever
  threshold $(j, 1)$ is removed from the queue. The \defn{size} of the
  queue $Q$ refers to the number of thresholds in the queue (rather
  than the number of cups).

\section{Technical Overview}\label{sec:technical}

In this section we give an overview of the analysis techniques used in
the paper. We begin by discussing the analysis of the asymmetric smoothed
greedy algorithm. To start, we focus our analysis on the
single-processor cup game, in which $p = 1$.

\paragraph{The unpredictability guarantee} 
At the heart of the analysis is what we call the
\defn{unpredictability guarantee}, which, roughly speaking,
establishes that the filler cannot predict large sets of cups that
will all be over-full at the same time as one-another. We show that if
an algorithm satisfies a certain version of the unpredictability guarantee, along
with certain natural ``greedy-like'' properties, then the algorithm is
guaranteed to exhibit a small tail size.

Formally, we say that an emptying algorithm satisfies
\defn{$R$-unpredictability} at a step $t$ if for any oblivious filling
algorithm, and for any set of cups $S$ whose size is a sufficiently
large constant multiple of $R$, there is high probability in $n$ that
at least one cup in $S$ has fill less than $1$ after step $t$. In
other words, for any polynomial $f(n)$, there exists a constant $c$
such that: for each set $S \subseteq [n]$ of $c R$ cups, the
probability that every cup in $S$ has height $1$ or greater at step
$t$ is at most $1/f(n)$.

\paragraph{How $R$-unpredictability helps}
Rather than proving that $R$-unpredictability causes the tail size to
stay small, we instead show the contrapositive. Namely, we show that
if there is a filling strategy that achieves a large tail size, the
strategy can be adapted to instead violate $R$-unpredictability.

Suppose that the filler is able to achieve tail size $cR$ at some step
$t$, where $c$ is a large constant. Then during each of the next $cR$
steps, the emptier will remove water from cups containing fill $2$ or
more (here, we use the crucial fact that the emptier always
prioritizes cups with fills $2$ or greater over cups with fills
smaller than $2$). This means that, during steps
$t + 1, \ldots, t + cR$, the set of cups with fill $1$ or greater is
monotonically increasing. During these steps the filler can place $1$
unit of water into each of the cups $1, 2, \ldots, cR$ in order to
ensure that these cups all contain fill $1$ or greater after step
$t + cR$. Thus the filler can transform the initial tail size of $cR$
into a large set of cups $S = \{1, 2, \ldots, cR\}$ that all have fill
$1$ or greater. In other words, any filling strategy for achieving
large tail size (at some step $t$) can be harnessed to violate
$R$-unpredictability (at some later step $t + cR$).

The directness of the argument above may seem to suggest that in order
to prove the $R$-unpredictability, one must first (at least
implicitly) prove a bound on tail size. A key insight in this paper is
that the use of priorities in the asymmetric smoothed greedy
algorithm allows for $R$-unpredictability to be analyzed as its own
entity.



  Our algorithm analysis establishes $\log n \log \log
  n$-unpredictability for the first $\poly n$ steps of any cup game,
  with high probability in $n$. This, in turn, implies a bound of
  $O(\log n \log \log n)$ on tail size.

  \paragraph{Establishing unpredictability}
  We prove that, out of the roughly $n / 2$ cups $j$ with priorities
  $p_j \ge 1/2$, at most $O(\log n \log \log n)$ of them are queued
  (i.e., contain fill $1$ or greater) at a time, with high probability
  in $n$. Recall that the cups with priority $p_j \ge 1/2$ are
  prioritized by the asymmetric smoothed greedy algorithm when the
  algorithm is choosing between cups with fills in the interval
  $[1, 2)$. This preferential treatment does not extend the case where
  there are cups containing fill $\ge 2$, however. Remarkably, the
  limited preferential treatment exhibited by the algorithm is enough
  to ensure that the number of queued high-priority cups never exceeds
  $O(\log n \log \log n)$.

  The bound of $O(\log n \log \log n)$ on the number of queued cups
  with priorities $\ge 1/2$ implies $\log n \log \log
  n$-unpredictability as follows. For any fixed set $S$ of cups, the
  number of cups $j$ in $S$ with priority $p_j \ge 1/2$ will be
  roughly $|S| / 2$ with high probability in $n$. If $|S| / 2$ is at
  least a sufficiently large constant multiple of $\log n \log \log
  n$, then the number of cups with $p_j \ge 1/2$ in $S$ exceeds the
  \emph{total} number of cups with $p_j \ge 1/2$ that are queued. Thus
  $S$ must contain at least one non-queued cup, as required for the
  unpredictability guarantee.

  In order to bound the number of queued cups with priority
  $p_j \ge 1/2$ by $O(\log n \log \log n)$, we partition the cups into
  $\Theta(\log \log n)$ \defn{priority levels} based on their
  priorities $p_j$. Let $q$ be a sufficiently large constant multiple
  of $\log \log n$.  The priority level of a cup $j$ is given by
  $\lfloor p_j \cdot q \rfloor + 1$. (Note that the priority levels are
  only needed in the analysis, and the algorithm does not have to know
  $q$.) We show that with high probability in $n$, there are never
  more than $O(\log n \log \log n)$ queued cups with priority level
  $\ge q/2$. Note that, although we only care about bounding the
  number of queued whose priority-levels are in the top fifty
  percentile, our analysis will take advantage of the fact that the
  priorities $p_j$ are defined at a high granularity (rather than, for
  example, being boolean).
    
   \paragraph{The stalled emptier problem}
   Bounding the number of queued cups with priority level greater than
   $\ell$ directly is difficult for the following reason: Over the
   course of a sequence of steps, the filler may cross many
   \emph{light} thresholds cups with priority level greater than
   $\ell$, while the emptier only removes \emph{heavy} thresholds from
   $Q$ (i.e., the emptier empties exclusively from cups of height $2$
   or greater). This means that, in a given sequence of steps, the
   number of queued cups with priority level greater than $\ell$ could
   increase substantially. We call this the \defn{stalled emptier
     problem}. Note that the stalled emptier problem is precisely what
   enables the connection between tail size and $R$-unpredictability
   above, allowing the filler to transform large tail size into a
   violation of $R$-unpredictability. As a result, any analysis that
   directly considers the stalled emptier problem must also first
   bound tail-size, bringing us back to where we started.

   Rather than bounding the number of queued cups with priority level
   greater than $\ell$, we instead compare the number of queued cups
   at priority level greater than $\ell$ to the number at priority
   level $\ell$. The idea is that, if the stalled-emptier problem
   allows for the number of queued priority-level greater than $\ell$
   to grow large, then it will allow for the number of queued
   priority-level-$\ell$ cups to grow even larger. That is, without
   proving any absolute bound on the number of cups at a given
   priority level, we can still say something about the ratio of
   high-priority cups to low-priority cups in the queue.

   To be precise, we prove that, whenever there are $k$ queued cups at
   some priority level $\ell$, there are at most
   $O(\sqrt{q k \log n} + \log n)$ queued cups at priority level
   $> \ell$ (recall that $q = \Theta(\log \log n)$ is the number of
   priority levels). Since the number of cups with priority level at
   least $1$ is deterministically anchored at $n$, this allows for us
   to inductively bound the number of queued cups with large priority
   levels $\ell$.  In particular, the number of queued cups at
   priority level $q/ 2$ or greater never exceeds
   $O(\log n \log \log n)$.

   \paragraph{Comparing the number of queued cups with priority level $\ell$ versus $> \ell$}
   Suppose that after some step $t$, there are some large number $k$
   of queued cups with priority level $\ge \ell$. We wish to show that
   almost all of these $k$ cups have priority level exactly $\ell$.

   Before describing our approach in detail (which we do in the
   following two subheaders), we give an informal description of the
   approach. Let $k_1$ be number of priority-level-$\ell$ queued cups,
   and let $k_2$ be the number of priority-level-greater-than-$\ell$
   queued cups. The only way that there can be a large number $k_2$ of
   priority-level-greater-than-$\ell$ cups queued is if they have all
   entered the queue since the last time that a level-$\ell$ cup was
   dequeued. This means that the size of $Q$ has increased by at least
   $k_2$ since the last time that a priority-level-$\ell$ cup was
   dequeued. On the other hand, we show that priority-level-$\ell$
   cups accumulate in $Q$ at a much faster rate than the size of $Q$
   varies. In particular, we show that both the rate at which
   priority-level $\ell$ cups accumulate in $Q$ and the rate at which
   $Q$'s size varies are controlled by what we call the ``influence''
   of a time-interval, and that the former is always much larger than
   the latter. This ensures that $k_1 \gg k_2$.
   
   Note that the analysis avoids arguing directly the number of queued
   high-priority cups small, which could be difficult due to the
   stalled emptier problem. Intuitively, the analysis instead shows
   that low priority cups do a good job ``pushing'' the high priority
   cups out of the queue, ensuring that the ratio of low-priority cups
   (i.e., cups with priority level $\ell$) to high-priority cups
   (i.e. cups with priority level $> \ell$) is always very large.

   \paragraph{Relating the number of high-priority queued cups to changes in $|Q|$}
    Let $t_0$ be the most recent step $t_0 \le t$ such that at least
    $k + 1$ \emph{distinct} cups $C$ with priority level $\ell$ cross
    thresholds during steps $t_0, \ldots, t$. (Recall that $k$ is the
    number of queued cups with priority level $\ge \ell$ after step
    $t$.) One can think of the steps $t_0, \ldots, t$ as representing
    a long period of time in which many cups with priority level
    $\ell$ have the opportunity to accumulate in $Q$. We will now show
    that the use of priorities in the asymmetric smoothed greedy
    algorithm causes the following property to hold: The number of
    queued cups with priority level $> \ell$ after step $t$ is bounded
    above by the amount that $|Q|$ varies during steps $t_0, \ldots,
    t$.
   
    Because $Q$ contains only $k$ queued cups with priority level
    $\ge \ell$ after step $t$, at least one cup from $C$ must be
    dequeued during steps $t_0, \ldots, t$ (otherwise, $Q$ would
    contain at least $|C| = k + 1$ level-$\ell$ cups after step
    $t$). Let $t^*$ be the final step in $t_0, \ldots, t$ out of those
    that dequeue a cup with priority level $\le \ell$, and let
    $Q_{t^*}$ and $Q_t$ denote the queue after steps $t^*$ and $t$,
    respectively.

   By design, the only way that the asymmetric smoothed greedy
   algorithm can dequeue a cup with priority level $\le \ell$ at step
   $t^*$ is if the queue $Q_{t^*}$ consists exclusively of light
   thresholds (i.e., thresholds of the form $(j, 1)$) for cups $j$
   with priority level $\le \ell$. Moreover, the thresholds in
   $Q_{t^*}$ must remain present in $Q_t$, since by the definition of
   $t^*$ no cups with priority level $\le \ell$ are dequeued during
   steps $t^* + 1, \ldots, t$.

   Since $Q_{t^*} \subseteq Q_t$ and $Q_{t^*}$ contains only
   thresholds for cups with priority level $\le \ell$, the total number of
   thresholds in $Q_t$ for cups with priority level $> \ell$ is at
   most $|Q_t| - |Q_{t^*}|$. In other words, the only way that a large
   number of cups with priority level $> \ell$ can be queued after
   step $t$ is if the size of $Q$ varies by a large amount during
   steps $t_0, \ldots, t$.

   Although $t - t_0$ may be very large compared to $k$ (e.g. $\poly
   n$) we show that the amount by which $|Q|$ varies during steps
   $t_0, \ldots, t$ is guaranteed to be small as a function of $k$,
   bounded above by $O(\sqrt{kq\log n})$. This means that, out of the
   $k$ cups with priority level $\ge \ell$ in $Q_t$, at most
   $O(\sqrt{kq \log n})$ of them can have priority level $\ell + 1$ or larger.

   \paragraph{The influence property: bounding the rate at which $|Q|$ varies}
   The main tool in order to analyze the rate at which $Q$'s size
   varies is to analyze sequences of steps based on their
   \defn{influence}. For sequence of steps $I$, the influence of $I$
   is defined to be $\sum_{j = 1}^n \min(1, c_j(I))$, where $c_j(I)$
   is the amount of water poured into each cup $j$ during interval
   $I$. We show that, for any priority level $\ell$ and for any step
   interval $I$ with influence $2 rq$ for some $r$, either
   $r = O(\log n)$, or two important properties are guaranteed to hold
   with high probability:
   \begin{itemize}
   \item \textbf{Step interval $I$ crosses thresholds in at least
       $r$ cups with priority level $\ell$.} This is true of any
     interval $I$ with influence at least $2qr$ by a simple
     concentration-bound argument.
   \item \paragraph{The size of $Q$ varies by at most $O(\sqrt{q r
       \log n})$ during step interval $I$} The key here is to show
     that, during each subinterval $I' \subseteq I$, the number of
     thresholds crossed by the filler is within $O(\sqrt{q r \log n})$
     of $|I'|$. In order to do this, we take advantage of the initial
     random offsets $r_j$ that are placed into each cup by the
     algorithm. If the filler puts some number $c_j(I')$ of units of
     water into a cup $j$ during $I'$, then the cup $j$ will
     deterministically cross $\lfloor c_j(I') \rfloor$ thresholds, and
     with probability $c_j(I') - \lfloor c_j(I') \rfloor $ will cross
     one additional threshold (with the outcome depending on the
     random value $r_j$). Since the influence of $I'$ is at most
     $2rq$, we know that $\sum_j (c_j(I') - \lfloor c_j(I') \rfloor)
     \le 2rq$. That is, if we consider only the threshold crossings
     that are not certain, then the number of them is a sum of
     independent 0-1 random variables with mean at most $2rq$. By a
     Chernoff bound, this number varies from its mean by at most
     $O(\sqrt{q r \log n})$, with high probability in $n$.
   \end{itemize}
   Combined, we call these the \defn{influence property}. By a union
   bound, the influence property holds with high probability on all
   sub-sequences of steps during the cup game, and for all values $r$.
   
   The influence property creates a link between the number of cups
   with priority level $\ell$ that cross thresholds during a sequence
   of steps $I$, and the amount by which $|Q|$ varies during steps
   $I$. Applying this link with $r = k + 1$ to steps $t_0, \ldots, t$,
   as defined above, implies that $|Q|$ varies by at most
   $O(\sqrt{qk \log n})$ during steps $t_0, \ldots, t$. This, in turn,
   bounds the number of queued cups with priority level $\ell + 1$ or
   larger by $O(\sqrt{qk \log n})$ after step $t$, completing the
   analysis.

   \paragraph{Extending the analysis to the multi-processor cup game}
   The primary difficulty in analyzing the multi-processor cup game
   (i.e., when $p > 1$) is that the emptier must remove water from $p$
   different cups, even if almost all of the water in the system
   resides in fewer than $p$ cups. For example, the emptier may
   dequeue a cup $j$ even though there are up to $p - 1$ other
   higher-priority cups that are still queued; furthermore, each of
   these higher-priority cups may contribute a large number of heavy
   thresholds to the queue $Q$.

   We solve this issue by leveraging recent bounds on backlog for
   the $p$-processor cup game \cite{Kuszmaul20}, which prove that the
   deterministic greedy emptying algorithm achieves backlog $O(\log
   n)$. This can be used to ensure that, for any $p - 1$ cups that are
   queued, each of them can only contribute a relatively small number
   of thresholds to the queue $Q$. These ``miss-behaving'' thresholds
   can then be absorbed into the algorithm analysis.
  
   \paragraph{Nearly matching lower bounds on tail size}
   Our lower-bound constructions extend the techniques used in past
   works for backlog \cite{Kuszmaul20, BenderFaKu19, DietzRa91} in
   order to apply similar ideas to tail size. One of the surprising
   features of our lower bounds is that they continue to be nearly
   tight even in the multi-processor case --- the same is not known to
   be true for backlog. We defer further discussion of the lower
   bounds to Section \ref{sec:lowerbounds}.

   \paragraph{Lower bounds against unending guarantees}
   Finally, we consider the question of whether the analysis of the
   asymmetric smoothed greedy algorithm can be extended to offer
   an \defn{unending guarantee}, i.e., a guarantee that for any step $t$,
   no matter how large, there a high probability at step $t$ that the
   backlog and tail size are small. 

   We show that, without the use of resource augmentation, unending
   guarantees are not possible for the asymmetric smoothed greedy
   algorithm, or, more generally, for any monotone memoryless emptying
   algorithm. Lower bounds against unending guarantees have previously
   been shown for the multi-processor cup game \cite{Kuszmaul20}, but
   remained open for the single-processor cup game.

The filling strategy, which we call the \defn{fuzzing algorithm}, has
a very simple structure: the filler spends a large number (i.e.,
$n^{\tilde{\Theta}(n)}$) of steps randomly placing water in multiples
of $1/2$ into cups $1, 2, \ldots, n$. The filler then disregards a
random cup, which for convenience we will denote by $n$, and spends a
large number of steps randomly placing water into the remaining cups
$1, 2, \ldots, n - 1$. The filler then disregards another random cup,
which we will call cup $n - 1$ ,and spends a large number of steps
randomly placing water into cups $1, 2, \ldots, n - 2$, and so on. We
call the portion of the algorithm during which the filler is focusing
on cups $1, 2, \ldots, i$ the \defn{$i$-cup phase}.

Rather than describe the analysis of the fuzzing algorithm (which is
somewhat complicated), we instead give an intuition for why the
algorithm works. For simplicity, suppose the emptier follows the
(standard) smoothed greedy emptying algorithm.

Between the $i$-cup phase and the $(i - 1)$-cup phase, the filler
disregards a random cup (that we subsequently call cup
$i$). Intuitively, at the time that cup $i$ is discarded, there is a
roughly $50\%$ chance that cup $i$ has more fill than the average of
cups $1, 2, \ldots, i$. Then, during the $(i - 1)$-cup phase, there is
a reasonably high probability that the filler \emph{at some point}
manages to make all of cups $1, 2, \ldots, i - 1$ have almost equal
fills to one-another. At this point, the emptier will choose to empty
out of cup $i$ instead of cups $1, 2, \ldots, i - 1$. The fact that
the emptier neglects cups $1, 2, \ldots, i - 1$ during the step, even
though the filler places $1$ unit of water into them, causes their
average fill to increase by $1 / (i - 1)$. Since this
happens with constant probability in every phase, the result is that,
by the beginning of the $\sqrt{n}$-cup phase there are $\sqrt{n}$ cups
each with expected fill roughly
$$\Omega\left(\frac{1}{n} + \frac{1}{n - 1} + \cdots + \frac{1}{\sqrt{n} + 1}\right) = \Omega(\log n).$$

Formalizing this argument leads to several interesting technical
problems. Most notably, the cups $1, 2, \ldots, i - 1$ having
\emph{almost} equal fills (rather than exactly equal fills) may not be
enough for cup $i$ to receive the emptier's attention. Moreover, if we
wish to analyze the \emph{asymmetric} smoothed greedy algorithm or,
more generally, the class of monotone memoryless algorithms, then cups
are not necessarily judged by the emptying algorithm based on their
fill heights, and may instead be selected based on an essentially
arbitrary objective function that need not treat cups
symmetrically. These issues are handled in Section \ref{sec:unending}
by replacing the notion of cups $1, 2, \ldots, i - 1$ having almost
equal fills as each other with the notion of cups $1, \ldots, i - 1$
reaching a certain type of specially designed equilibrium state that
interacts well with the emptier.

   \section{Algorithm analysis}\label{sec:analysis}

   In this section, we give the full analysis of the $p$-processor
   asymmetric smoothed greedy algorithm. The main result of the
   section is Theorem \ref{thm:main}, which bounds the tail size of
   the game by $O(\log n \log \log n + p \log p)$ for the first $\poly
   n$ steps of the game with high probability in $n$.
   
   In addition to using the conventions from Section \ref{sec:alg} we
   find it useful to introduce one additional notation: for a sequence
   of steps $I$, define $c_j(I)$ to be the amount of water placed into
   cup $j$ during $I$. We also continue to use the convention from
   Section \ref{sec:technical} that $q$ is a large constant multiple
   of $\log \log n$, and that each cup $j$ is assigned a
   \defn{priority level} given by $\lfloor p_j \cdot q\rfloor +
   1$. 

   Recall that a cup $j$ crosses a threshold $(j, i)$ whenever the
   fill of cup $j$ increases from some quantity $f < i$ to some
   quantity $f' \ge i$ for $i \in \mathbb{N}$. A key property of the
   smoothed greedy algorithm, which was originally noted by Bender et
   al. \cite{BenderFaKu19}, is that the number of threshold crossings
   across any sequence of steps can be expressed using a sum of
   independent 0-1 random variables.\footnote{Note that, when counting
     the number of threshold crossings across a sequence of steps,
     the same threshold $(j, i)$ may get crossed multiple times, and
     thus contribute more than $1$ to the count.}  This remains true
   for the asymmetric smoothed greedy algorithm, and is formalized
   in Lemma \ref{lem:smoothing}.
   \begin{lemma}[Counting threshold crossings]
   For a sequence of steps $I$, and for a cup $j$, the number of
   threshold crossings in cup $j$ is $\lfloor c_j(I) \rfloor + X_j$,
   where $X_j$ is a 0-1 random variable with mean $c_j(I) - \lfloor
   c_j(I) \rfloor$. Moreover, $X_1, X_2, \ldots, X_n$ are independent.
   \label{lem:smoothing}
   \end{lemma}
   \begin{proof}
     Recall that the emptier only removes water from a cup $j$ if cup
     $j$ contains at least $1$ unit. Moreover, the emptier always
     removes exactly $1$ unit of water from cups. Since threshold
     crossings in each cup $j$ depend only on the fractional amount of
     water (i.e., the amount of water modulo $1$) in the cup, the
     behavior of the emptier cannot affect when thresholds are crossed
     within each cup.

     Let $t_0$ be the final step prior to interval $I$. For each cup
     $j$, the fractional amount of water in the cup at the beginning
     of interval $I$ is
     \begin{equation}
       r_j + c_j([1, t_0]) \bmod 1.
       \label{eq:randomfractional}
     \end{equation}

     Since $r_j$ is uniformly random in $[0, 1]$, it follows that
     \eqref{eq:randomfractional} is as well. The first $c_j(I) -
     \lfloor c_j(I) \rfloor$ units of water poured into cup $j$ during
     interval $I$ will therefore cross a threshold with probability
     exactly $c_j(I) - \lfloor c_j(I) \rfloor$. The next $\lfloor
     c_j(I) \rfloor$ units of water placed into cup $j$ are then
     guaranteed to cause exactly $\lfloor c_j(I) \rfloor$ threshold
     crossings. The number of crossings in cup $j$ during the step
     sequence is therefore $\lfloor c_j(I) \rfloor + X_j$, where $X_j$
     is 0-1 random variable with mean $c_j(I) - \lfloor c_j(I)
     \rfloor$, and where the randomness in $X_j$ is due to the random
     initial offset $r_j$. Because $r_1, r_2, \ldots, r_n$ are
     independent, so are $X_1, X_2, \ldots, X_n$.
   \end{proof}

   One consequence of Lemma \ref{lem:smoothing} is that, if a sequence
   of steps $I$ has a large influence $s$, then each priority level
   $\ell$ will have at least $\Omega(s / q)$ cups that cross
   thresholds during interval $I$ (recall that $q$ is the number of
   priority levels).
   \begin{lemma}[The influence property, part 1]
     Consider a sequence of steps $I$ with influence $s$. Let $\ell$
     be a priority level. With high probability in $n$, at least
     $\frac{s}{2q} - O(\log n)$ distinct cups with priority level
     $\ell$ cross thresholds during the sequence of steps $I$.
     \label{lem:lots-of-k-crossings}
   \end{lemma}
   \begin{proof}
     By Lemma \ref{lem:smoothing}, the probability that cup $j$
     crosses at least one threshold during step sequence $I$ is
     $\min(c_j(I), 1)$, independently of other cups $j'$. The number
     $X$ of distinct cups that cross thresholds during interval $I$ is
     therefore a sum of independent indicator random variables with
     mean $s$, where $s$ is the influence of $I$. Since each cup has
     probability $\frac{1}{q}$ of having priority level $\ell$, the
     number $Y$ of cups with priority level $\ell$ that cross
     thresholds in interval $I$ is a sum of independent indicator
     random variables with mean $\frac{s}{q}$.

     If $s / q \le O(\log n)$, the number of distinct cups with
     priority level $\ell$ to cross thresholds is at least $0 \ge \frac{s}{2q} - O(\log n)$ trivially. Suppose, on the other hand, that $s
     / q \ge c \log n$ for a sufficiently large constant $c$. Then by
     a Chernoff bound,
     $$\Pr\left[Y < \frac{s}{2q}\right] \le \exp\Big[- \frac{s}{8q}  \Big] \le \frac{1}{n^{c / 8}},$$
     completing the proof.
   \end{proof}

   The proofs of the preceding lemmas have not needed to explicitly
   consider the effect of there being a potentially large number $p$
   of processors. In subsequent proofs, the multi-processor case will
   complicate the analysis in two ways. First, the emptier may
   sometimes dequeue a cup, even when there are more than $p$ heavy
   thresholds in the queue (this can happen when the heavy thresholds
   all belong to a set of fewer than $p$ cups). Second, and similarly,
   the emptier may sometimes be unable to remove a full $p$ thresholds
   from the queue $Q$, even though $|Q| > p$ (this can happen if of
   the thresholds in $Q$ belong to a set of fewer than $p$ cups). It
   turns out that both of these problems can be circumvented using the
   fact that the emptying algorithm achieves small backlog. In
   particular, this ensures that no single cup can ever contribute
   more than $O(\log \log n + \log p)$ thresholds to $Q$:
   \begin{lemma}[K. \cite{Kuszmaul20}]
     In any multi-processor cup game of $\poly n$ length, the
     asymmetric smoothed greedy algorithm achieves backlog $O(\log
     \log n + \log p)$ after each step, with high probability in $n$.
     \label{lem:Kuszmaul20}
   \end{lemma}
   \begin{proof}
     This follows from Theorem 5.1 of \cite{Kuszmaul20}. Although
     \cite{Kuszmaul20} considers the smoothed greedy algorithm (rather
     than the asymmetric smoothed greedy algorithm), the analysis
     applies without modification.
   \end{proof}
   
   Using Lemma \ref{lem:Kuszmaul20} as a tool to help in the case of
   $p > 1$, we now return to the analysis approach outlined in Section
   \ref{sec:technical}.

   \begin{rem}
     The proof of Lemma \ref{lem:Kuszmaul20} given in
     \cite{Kuszmaul20} is highly nontrivial. We remark that, although
     Lemma \ref{lem:Kuszmaul20} simplifies our analysis, there is also
     an alternative lighter weight approach that one can use in place
     of the lemma. In particular, one can begin by analyzing the
     \defn{$h$-truncated cup game} for some sufficiently large $h \le
     O(\log \log n + \log p)$. In this game, the height of each cup is
     deterministically bounded above by $h$, and whenever the height
     of a cup exceeds $h$, $1$ unit of water is removed from the cup
     (and that unit does not count as part of the emptier's turn). The
     $h$-truncated cup game automatically satisfies the backlog
     property stated by Lemma \ref{lem:Kuszmaul20}, allowing for it to
     be analyzed without requiring the lemma. The analysis can then be
     used to bound the backlog of the $h$-truncated cup game to at
     most $h / 2$ with high probability (using the analysis by
     \cite{Kuszmaul20} of the greedy algorithm, applied to the
     $\tilde{O}(\log n + p)$ cups in the tail). It follows that with
     high probability, the $h$-truncated cup game and the (standard)
     cup game are indistinguishable. This means that the
     high-probability bounds on tail size for the $h$-truncated cup
     game carry over directly to the the standard cup game.
   \end{rem}
   
   The next lemma shows that, even though many threshold crossings may
   occur in a sequence of steps $I$, the size of the queue $Q$ varies by only
   a small amount as a function of the influence of $I$.

   \begin{lemma}[The influence property, part 2]
     Consider a sequence of steps $I$ with influence at most $s$
     during a game of length at most $\poly n$. For each step $t \in
     I$, let $Q_t$ denote the queue after step $t$. With
     high probability in $n$,
     $$\Big||Q_{t_1}| - |Q_{t_2}|\Big| \le O(\sqrt{s \log n} + \log n + p (\log \log n + \log p))$$
          for all subintervals $[t_1, t_2] \subseteq I$.
     \label{lem:bolus-variance}
   \end{lemma}
   \begin{proof}

     We begin with a simpler claim:
     \begin{claim}
       For any subinterval $I' \subseteq I$, the number of threshold
       crossings during $I'$ is within $O(\sqrt{s \log n} + \log
       n)$ of $p|I'|$ with high probability in $n$.
       \label{clm:threshold-variance}
     \end{claim}
     \begin{proof}
     Because $I$ has influence at most $s$ so does $I'$. Lemma
     \ref{lem:smoothing} tells us that, during $I'$, the number of
     threshold crossings $X$ that occur satisfies $\E[X] =
     p|I'|$. Moreover, $X$ satisfies $X = A + \sum_{j = 1}^n X_j$,
     where $A$ is a fixed value and the $X_j$'s are independent 0-1
     random variables, each taking value $1$ with probability $c_j(I')
     - \lfloor c_j(I') \rfloor \le \min(c_j(I'), 1)$. Note that $I'$
     has influence $\sum_j \min(c_j(I'), 1) \le s$, and thus
     $\E\left[\sum_{j = 1}^n X_j\right] \le s$. By a multiplicative
     Chernoff bound, it follows that for $\delta < 1$,
     $$\Pr[|X - \E[X]| \ge \delta s] \le 2\exp\left[-\delta^2 s / 3\right].$$

     Set $\delta = c \sqrt{\log n} / \sqrt{s}$ for a sufficiently
     large constant $c$. If $\delta > 1$, then $s \le O(\log n)$ and
     $|X - \E[X]|$ is deterministically $O(\log n)$. Otherwise,
     $$\Pr[|X - \E[X]| \ge c \sqrt{s \log n}] \le 2\exp\left[- c^2
       \log n / 3\right] \le \frac{1}{\poly n}.$$ Since $\E[X] =
     p|I'|$, it follows that the number of threshold crossings in
     interval $I'$ is within $O(\sqrt{s \log n} + \log n)$ of
     $p|I'|$ with high probability in $n$.
     \end{proof}

     Applying a union bound to the $\poly n$ subintervals of steps $I'
     \subseteq I$, Claim \ref{clm:threshold-variance} tells us that
     every subinterval $I' \subseteq I$ contains $p|I'| \pm O(\sqrt{s
       \log n} + \log n)$ threshold crossings with high probability in
     $n$.

     To complete the proof, consider some subinterval $I' \subseteq I$
     and let $m$ be the (absolute) amount by which $Q$ changes in size
     during $I'$. We wish to show that $m \le O(\sqrt{s \log n} + \log n + p (\log \log n + \log p))$.

     Suppose that $|Q|$ shrinks by $m$  during $I'$. Then the number of threshold crossings in
     subinterval $I'$ would have to be at most $p|I'| - m$, meaning
     that $m \le O(\sqrt{s \log n} + O(\log n))$, as desired.

     
     Suppose, on the other hand, that $|Q|$ grows by $m$ during
     $I'$. Call a step $t \in I'$ \defn{removal-friendly} if the
     emptier removes $p$ full units of water during step $t$ (i.e.,
     prior to the emptier removing water, there are at least $p$ cups
     with height $1$ or greater). By Lemma \ref{lem:Kuszmaul20}, with
     high probability in $n$, the size of $Q$ after any
     removal-unfriendly step is at most $O(p (\log \log n + \log
     p))$. If $I'$ consists exclusively of removal-friendly steps,
     then the filler must cross at least $p |I'| + m$ threshold
     crossings in order to increase $|Q|$ by $m$; thus $m \le
     O(\sqrt{s \log n} + \log n)$. On the other hand, if $I'$ contains
     at least one removal-unfriendly step, then there must be some
     last such step $t$ in $I' = (t_0, t_1]$. Since $|Q| \le O(p (\log \log n + \log
     p))$ after step $t$ but $|Q| \ge m$ after step $t_1$, it must be that during the
  steps $(t, t_1]$ the size of $Q$ increases by at least $m - O(p
(\log \log n + \log p))$. Since the interval $(t, t_1)$ consists entirely of removal-friendly steps, we can apply the reasoning from the first case (i.e., the case of only removal-friendly steps) to deduce that $m \le O(\sqrt{s \log n} +
\log n + p(\log \log n + \log p))$, completing the proof.
   \end{proof}

   Combined, Lemmas \ref{lem:lots-of-k-crossings} and
   \ref{lem:bolus-variance} give the influence property discussed in
   Section \ref{sec:technical}. Using this property, we can now
   relate the number of queued cups with priority level $\ge \ell$ to
   the number of queued cups with priority level $\ge \ell + 1$ for a
   given priority level $\ell \in \mathbb{N}$.

   \begin{lemma}[Accumulation of low-priority cups]
     Let $t \le \poly n$, let $K$ be the number of queued cups with
     priority level $\ge \ell$ after step $t$, and $m$ be the number
     of queued cups with priority level $\ge \ell + 1$ after step
     $t$. With high probability in $n$,
     \begin{equation}
       m \le O(\sqrt{q K \log n} + \log n + p (\log \log n + \log p)).
       \label{eq:queued-thresholds}
     \end{equation}
     \label{lem:queued-thresholds}
   \end{lemma}
   \begin{proof}
     For each $k \in \{1, 2, \ldots, n\}$, define $I_k$ to be the
     smallest step-interval ending at step $t$ and with influence at
     least $2qk$ (or define $I_k = [1, t]$ if the total influence of
     $[1, t]$ is less than $2qk$). By Lemmas
     \ref{lem:lots-of-k-crossings} and \ref{lem:bolus-variance}, each
     $I_k$ satisfies the following two properties with high
     probability in $n$:
     \begin{itemize}
     \item \textbf{The many-crossings property.} Either $I_k$ is all of $[1, t]$, or the number of
       priority-level-$\ell$ cups to cross thresholds during $I_k$ is
       at least $k - O(\log n)$.
     \item \textbf{The low-variance property. }The size of $Q$ varies
       by at most $B_{k} := O(\sqrt{qk\log n} + \log n + p (\log \log
       n + \log p))$ during $I_k$. To see this, we use the fact that $I_k$ has
       influence at most $2qk + p$, which by Lemma
       \ref{lem:bolus-variance} limits the amount by which $Q$ varies
       to $$O(\sqrt{(qk + p)\log n} + \log n + p (\log \log n + \log
       p)).$$ Since $\sqrt{(qk + p)\log n} \le \sqrt{q k \log n} +
       \sqrt{p \log n} \le \sqrt{qk \log n} + p + \log n$, it follows
       that the amount by which $Q$ varies during $I_k$ is at most
       $O(\sqrt{qk\log n} + \log n + p (\log \log n + \log p))$, with
       high probability in $n$.
     \end{itemize}
     Collectively, this pair of properties is called the
     \defn{influence property}. By a union bound, the influence
     property holds for all $k \in \{1, 2, \ldots, n\}$ with high
     probability in $n$. It follows that the property also holds for
     $k = K$ (recall $K$ is the number of queued cups with priority
     level $\ge \ell$ after step $t$).
     
     If $I_{K} = [1, t]$, then the total size of $Q$ can be at most
     $B_{K}$ (since $Q$ begins as size $0$ at the start of
     $I_{K}$). It follows that $m \le
     B_{K}$, meaning that \eqref{eq:queued-thresholds} is
     immediate. In the rest of the proof, we focus on the case in
     which $I_{K} \neq [1, t]$.

     If the emptier never dequeues any priority-level-$\ell$ cups
     during $I_{K}$, then by the many-crossings property, there are at
     least $K - O(\log n)$ priority-level-$\ell$ cups queued after
     step $t$. The number of queued cups with priority levels greater
     than $\ell$ is therefore at most $O(\log n)$, as desired.

     Suppose, on the other hand, that there is at least one step in
     $I_K$ at which the emptier dequeues a priority-level-$\ell$ (or
     smaller) cup, and let $t^*$ be the last such step. Let $Q_{t^*}$
     be the set of queued thresholds after $t^*$, and let $Q_t$ be the
     set of queued thresholds after step $t$. Call a threshold in
     $Q_{t^*}$ \defn{permanent} if it is a light threshold for a cup
     with priority level $\le \ell$. All permanent thresholds in
     $Q_{t^*}$ are guaranteed to also be in $Q_t$, since $t^*$ is the
     final step in $I_K$ during which the emptier dequeues such a
     threshold. The non-permanent thresholds in $Q_{t^*}$ must reside
     in a set of fewer than $p$ cups, since the emptier would rather
     have dequeued one of them during step $t^*$ than to have dequeued
     a cup with priority level $\le \ell$. By Lemma
     \ref{lem:Kuszmaul20}, the number of non-permanent thresholds in
     $Q_{t^*}$ is therefore at most $O(p\log\log n + p \log p)$, with
     high probability in $n$.

     By the low-variance property, the sizes of $Q_t$ and $Q_{t^*}$
     differ by at most $B_K$. It follows that the permanent thresholds
     in $Q_{t^*}$ make up all but
     $$B_K + O(p \log \log n + p \log p) \le O(B_K)$$ of the
     thresholds in $Q_{t}$. 

     Recall that $Q_t$ contains thresholds from $K$ different cups with
     priority level $\ge \ell$. It follows that $Q_{t^*}$ contains
     permanent thresholds from at least $K - O(B_K)$ different cups with
     priority level $\ge \ell$. The permanent thresholds in $Q_{t^*}$
     are all for cups with priority level $\le \ell$, however. Thus
     there are at least $K - O(B_K)$ cups with priority level $\ell$
     that are queued after step $t^*$ and remain queued after step
     $t$. This bounds the number of queued cups after step $t$ with
     priority level greater than $\ell$ by at most $O(B_K)$,
     completing the proof.

   \end{proof}

   Since the number of queued cups with priority level $\ge 1$ can
   never exceed $n$, Lemma \ref{lem:queued-thresholds} allows for us
   to bound the number of queued cups with priority level $\ge \ell$
   inductively. We argue that if $q$ is a sufficiently large constant
   multiple of $\log \log n$, then the number of queued cups with
   priority level $\ge q/2$ never exceeds $O(\log n \log \log n + p (\log \log n +
   \log p))$, with high probability in $n$. This can then be used to
   obtain $(\log n \log \log n + p \log p)$-unpredictability, as
   defined in Section \ref{sec:technical}.
   \begin{lemma}[The unpredictability guarantee]
     Consider a cup game of length $\poly n$. For any step $t$, and
     for any set of cups $S$ whose size is a sufficiently large
     constant multiple of $\log n \log \log n + p \log p$, at least
     one cup in $S$ is not queued after step $t$, with high
     probability in $n$. In other words, each step $t$ in the game
     satisfies $(\log n \log \log n + p \log p)$-unpredictability.
     \label{lem:unpredictability}
   \end{lemma}
   \begin{proof}

     Suppose the number of priority levels $q$ is set to be a
     sufficiently large constant multiple of $\log \log n$. For $\ell
     \in \{1, 2, \ldots, q\}$, let $m_\ell$ denote the maximum number
     of queued cups with priority level $\ge \ell$ during the game. We
     claim that $m_{q/2} \le O(\log n \log \log n + p \log p)$ with
     high probability in $n$.
     
     By Lemma \ref{lem:queued-thresholds}, for any $1 \le \ell < q$,
     \begin{equation*}
       m_{\ell + 1} \le O(\sqrt{q m_\ell \log n} + \log n + p \log
       \log n + p \log p),
     \end{equation*}
     with high probability in $n$. If we define $X = q \log n + \log n
     + p \log \log n + p \log p$, then it follows that,
     \begin{equation}
       m_{\ell + 1} \le O(\sqrt{X m_\ell} + X).
       \label{eq:shrinkX}
     \end{equation}

     For each priority level $\ell$, let $\delta_{\ell}$ be the
     ratio
     $\delta_\ell = \frac{m_\ell}{X}$. By \eqref{eq:shrinkX}, for any $1 \le \ell < q$, either
     $\delta_{\ell + 1} \le O(1)$ or
     $$\delta_{\ell + 1} \le O(\sqrt{\delta_\ell}).$$ It follows that,
     as long as $q$ is a sufficiently large constant multiple of $\log
     \log n$, then $\delta_{q/ 2} \le O(1)$, and thus
     $m_{q/ 2} \le O(X)$.

     Now consider a set of cups $S$ of size at least $cX$, where $c$
     is a sufficiently large constant. By a Chernoff bound, the number
     of cups with priority level greater than $q / 2$ in $S$ is at
     least $cX / 4$, with high probability in $n$. Since $c$ is a
     sufficiently large constant, and since $m_{q/2} \le O(X)$, this
     implies that $S$ contains more than $m_{q/2}$ cups with priority
     level $q/2$ or greater. Thus the priority-level-$\ell$ cups in
     $S$ cannot all be queued after step $t$.

     Note that $X \le O(\log n \log \log n + p \log p)$. Thus every
     set $S$ whose size is a sufficiently large constant multiple of
     $\log n \log \log n + p \log p$ has high probability of
     containing at least one non-queued cup after step $t$, completing
     the proof of $\log n \log \log n + p \log p$-unpredictability.
   \end{proof}

   To complete the analysis of the algorithm, we must formalize the
   connection between the unpredictability guarantee and tail size.   

   \begin{lemma}
     Suppose that a (randomized) emptying algorithm for the $p$-processor cup game
     on $n$ cups satisfies $R$-unpredictability in the steps of any
     game of polynomial length, and further satisfies the
     ``greediness property'' that whenever there is a cup of height
     at least $2$, the algorithm empties out of such a cup. Then in
     any game of polynomial length, the tail size $s$ after each step
     $t$ is $O(R + p)$ with high probability in $n$.
     \label{lem:reduceproperty}
   \end{lemma}
   \begin{proof}
     Consider a polynomial $f \in \poly n$, let $c$ be a large constant,
     and let $t \le \poly n$. Suppose for contradiction that there is a
     filling strategy such that, at time $t$ the tail size is at least
     $cR + p$ with probability at least $1/f(t)$. Since the tail size is at
     least $cR + p$ at time $t$, then during each of steps $I = \{t + 1,
     \ldots, t + \lceil cR / p \rceil\}$, the emptier removes water
     exclusively from cups with fills at least $2$. This means that the
     set of cups containing $1$ or more units of water is monotonically
     increasing during steps $t + 1, \ldots, t + \lceil cR / p
     \rceil$. If the filler places $1$ unit of water into each cup $1, 2,
     \ldots, cR$ during steps $t + 1, \ldots, t + \lceil cR / p \rceil$,
     then it follows that each of cups $1, 2, \ldots, cR$ has fill $1$ or
     greater after step $t + \lceil cR / p \rceil$.
     
     The preceding construction guarantees that, with probability at
     least $1/f(x)$, all of cups $1, 2, \ldots, cR$ contain at least
     $1$ unit of water at step $t + \lceil cR / p \rceil$. If $c$ is a
     sufficiently large constant, this violates $R$-unpredictability
     at step $t + \lceil cR / p \rceil$, a contradiction.
   \end{proof}

   Using the unpredictability guarantee, we can now complete the
   algorithm analysis:
   
   \begin{theorem}
     Consider a $p$-processor cup game that lasts for $\poly n$ steps
     and in which the emptier follows the asymmetric smoothed greedy
     algorithm. Then with high probability in $n$, the number of cups
     containing $2$ or more or units of water never exceeds $O(\log n
     \log \log n + p \log p)$ and the backlog never exceeds $O(\log
     \log n + \log p)$ during the game.
     \label{thm:main}
   \end{theorem}
   \begin{proof}
     By Lemma \ref{lem:unpredictability}, $(\log n \log \log n + p
     \log p)$-unpredictability holds for each step in any game of
     length $\poly n$. By Lemma \ref{lem:reduceproperty}, it follows
     that the tail size remains at most $O(\log n \log \log n + p \log
     p)$, with high probability in $n$, during any game of length
     $\poly n$.

     Lemma \ref{lem:Kuszmaul20} bounds the height of the fullest cup
     in each step by $O(\log \log n + \log p)$ with high probability
     in $n$. Alternatively, by considering a cup game consisting only
     of the cups that contain greater than $2$ units of water, the
     analysis of the deterministic greedy emptying algorithm (see
     \cite{AdlerBeFr03} for $p = 1$ and \cite{Kuszmaul20} for $p > 1$)
     on $O(\log n \log \log n + p \log p)$ cups implies that no cup
     ever contains more than $O(\log \log n + \log p)$ water, with
     high probability in $n$.
   \end{proof}

   \section{Lower Bounds}\label{sec:lowerbounds}

   In this section we prove that the asymmetric smoothed greedy
   algorithm achieves (near) optimal tail size within the class of
   backlog-bounded algorithms.

   An emptying algorithm is \defn{backlog bounded} if the algorithm
   guarantees that the backlog never exceeds $f(n)$ for some polynomial
   $f$. This is a weak requirement in that that the greedy algorithm
   achieves a bound of $O(\log n)$ on backlog \cite{AdlerBeFr03,
     Kuszmaul20}.  The main result in this section states that any
   backlog-bounded emptying algorithm must allow for a tail size of
   $\tilde{\Omega}(\log n + p)$ with probability $\frac{1}{\poly
     n}$. The lower bound continues to hold, even when the
   height-requirement for a cup to be in the tail is increased to an
   arbitrarily large constant (rather than $2$). When $p = 1$, the
   lower bound also applies to non-backlog-bounded emptying algorithms (Lemma \ref{lem:lowerbound1}).

   \begin{theorem}
     Let $c_1$ be a constant, and suppose $n \ge p + c_2$ for
     sufficiently large constant $c_2$. For any backlog-bounded
     emptying strategy, there is a $\poly n$-step oblivious randomized
     filling strategy that gives the following guarantee. After the
     final step of the filling strategy, there are at least
     $\Omega(\log n / \log \log n + p)$ cups with fill $c_1$ or
     greater, with probability at least $\frac{1}{\poly n}$.
     \label{thm:lowerbound}
   \end{theorem}
   
   To prove Theorem \ref{thm:lowerbound}, we begin by describing a
   simple lower-bound construction that we call the $(p, k,
   c)$-filling strategy. The strategy is structurally similar to the
   lower-bound construction for backlog given by Bender et al
   \cite{BenderFaKu19}.

   \begin{lemma}
     Let $k , c \in \mathbb{N}$  such that $c \ge 2$, $k \le n$,
     and $\frac{k}{p e^c} \ge 2$. Then there exists an $O(k)$-step
     oblivious randomized filling strategy for the $p$-processor cup
     game on $n$ cups that causes $\Omega(k / e^c)$ cups to each have
     fill at least $\Theta(c)$, with probability at least
     $\frac{1}{k^k}$. We call this strategy the $(p, k, c)$-filling
     strategy.
     \label{lem:fillingstrategy}
   \end{lemma}
   \begin{proof}
   Define the \defn{$(p, k, c)$-filling strategy} for the
   $p$-processor cup game on $n \ge k$ cups as follows. In each step
   $i$ of the strategy, the filler places $\frac{p}{k - p(i - 1)}$
   units of water into each of $k - p(i - 1)$ cups. The sets of cups
   $S_i$ used in each step $i$ are selected so that $S_{i + 1} = S_i
   \setminus \{x_1, x_2, \ldots, x_p\}$ for some random distinct $x_1,
   x_2, \ldots, x_p \in S_i$. The $(p, k, c)$-filling strategy
   completes after $t$ steps where $t = \lfloor \frac{k}{p}(1 -
   e^{-c}) \rfloor - 1$. Note that $k - pi \ge p$ for every step $i$.

   We say that the $(p, k, c)$-filling strategy \defn{succeeds} if at
   the beginning of each step $i$ none of the cups in $S_i$ have been
   touched (i.e., emptied from) by the emptier. If the $(p, k,
   c)$-filling strategy succeeds, then at the end the $i$-th step of
   the strategy there will be $k - pi$ cups each with fill
   \begin{align*}
     & \frac{p}{k} + \frac{p}{k - p} + \cdots + \frac{p}{k - p(i - 1)} \\
     & = \Theta\left(\frac{1}{k} + \frac{1}{k - 1} + \cdots + \frac{1}{k - pi + 1}\right) \\
     & = \Theta\left(\log \frac{k}{k - pi}\right),   
   \end{align*}
   where the first equality uses the fact that $k - pi \ge p$.

   Now consider the final step $t$ of a successful $(p, k, c)$-filling
   strategy. By the requirement that $\frac{k}{pe^c} \ge 2$, 
   $$k - pt = k - p \Big\lfloor \frac{k}{p}(1 - 1/e^c) \Big\rfloor - p \ge
   \frac{k}{e^c} - p \ge \frac{k}{2e^c}.$$ It follows that, after step
   $t$ of a successful $(p, k, c)$-filling strategy, there are at
   least $\Theta(k / e^c)$ cups, each with fill at least
   $$\Omega\left( \log \frac{k}{k / (2e^c)}\right) = \Omega(c).$$

   Next we evaluate the probability of a $(p, k, c)$-filling strategy
   being successful. If the first $i$ steps of the $(p, k, c)$-filling
   strategy all succeed, then the $(i + 1)$-th step has probability at
   least $\frac{1}{k^p}$ of succeeding. In particular, the emptier may
   touch up to $p$ cups $j_1, \ldots, j_p \in S_i$ during step $i$,
   and then the set $S_{i + 1} = S_i \setminus \{x_1, \ldots, x_p\}$
   has probability at least $\frac{1}{k^p}$ of removing a superset of
   those cups from $S_i$ to get $S_{i + 1}$. Since there are at most
   $k/p$ steps, the $(p, k, c)$-filling strategy succeeds with
   probability at least $\frac{1}{k^k}$.
   \end{proof}

   If we assume that $\log n / \log \log n$ is a sufficiently large
   constant multiple of $p$, then we can apply Lemma
   \ref{lem:fillingstrategy} directly to achieve a tail size of size
   $\tilde{\Omega}(\log n)$. (Furthermore, note that Lemma
   \ref{lem:lowerbound1} does not have any requirement that the
   emptier be backlog-bounded.)
   \begin{lemma}
     Let $c_1$ be a positive constant, and suppose $p \le \frac{\log
       n}{c_2\log \log n}$ for a sufficiently large constant $c_2$
     (where $c_2$ is large relative to $c_1$). Then there is an
     $O(\log n / \log \log n)$-step oblivious randomized filling
     strategy for the $p$-processor cup game on $n$ cups that causes
     $\Omega(\log n)$ cups to all have height $c_1$ or greater after
     some step $t \le \poly n$ with probability at least
     $\frac{1}{\poly n}$.
     \label{lem:lowerbound1}
   \end{lemma}
   \begin{proof}
     Let $c \in \mathbb{N}$ be a sufficiently large constant
     compared to $c_1$. By assumption that $c_2$ is sufficiently large
     in terms of $c_1$, we may also assume that
     \begin{equation}
       \frac{\log n / \log \log n}{pe^{c}} \ge 2.
       \label{eq:c0c1}
     \end{equation}
     
     By \eqref{eq:c0c1}, we can use Lemma \ref{lem:fillingstrategy} to
     analyze the $(p, \log n / \log \log n, c)$-filling
     strategy. The strategy causes
     $$\Omega\left(\frac{\log n}{\log \log n} \cdot e^{-c}
     \right) \ge \Omega(\log n / \log \log n)$$ cups to all have
     height at least $c_1$ with probability at least
     $$\left(\log n / \log \log n\right)^{-\log n / \log \log n} \ge
     \frac{1}{\poly n}.$$
   \end{proof}

   The next lemma gives a filling strategy for achieving tail size
   $\Omega(p)$ against any backlog-bounded emptying
   strategy. Remarkably, the construction in Lemma
   \ref{lem:lowerbound2} succeeds with probability $1 -
   e^{-\Omega(p)}$ (rather than with probability $1/\poly n$).
   
   \begin{lemma}
     Let $c_1$ be a constant, and suppose $n \ge p + c_2$ for
     sufficiently large constant $c_2$. For any backlog-bounded
     emptying strategy, there is a $\poly n$-step oblivious randomized
     filling strategy that gives the following guarantee. After the
     final step of the filling strategy, there are at least
     $\Omega(p)$ cups with fill $c_1$ or greater, with probability $1
     - e^{-\Omega(p)}$.
     \label{lem:lowerbound2}
   \end{lemma}
   \begin{proof}
     For the sake of simplicity, we allow for the filler to sometimes
     \defn{swap} two cups, meaning that the labels of the cups are
     interchanged.
 
     The basic building block of the algorithm is a \defn{mini-phase},
     which consists of $O(1)$ steps. In each step of a mini-phase the
     filler places $1$ unit of water into each of cups $1, 2, \ldots,
     p - 1$, and then strategically distributes $1$ additional unit of
     water among cups $p , p + 1, \ldots, n$. Using the final unit of
     water, the filler follows a $(1, c e^{c}, c)$-filling strategy on
     cups $p, p + 1, \ldots, n$, where $c$ is a sufficiently large
     constant relative to $c_1$ satisfying $n \ge p + c e^{c}$. We say
     that a mini-phase \defn{succeeds} if the emptier removes only $1$
     unit of water from cups $\{p, p + 1, \ldots, n\}$ during each
     step in the mini-phase, and the $(1, ce^c, c)$-filling strategy
     succeeds within the mini-phase. By the Lemma
     \ref{lem:fillingstrategy}, any successful mini-phase will cause
     at least one cup $j$ to have fill at least $c_1$ at the end of
     the mini-phase (and the filler will know $j$).

     Mini-phases are composed together by the filler to get
     \defn{phases}. During each $i$-th phase, the filler selects a
     random $w_i \in [1, f(n)n^2]$ and performs $w_i$ mini-phases
     (recall that $f(n)$ is the polynomial such that the emptier
     achieves backlog $f(n)$ or smaller). After the $w_i$-th
     mini-phase, the filler swaps cups $i$ and $j$, where $i$ is the
     phase number and $j$ is the the cup containing fill $\ge c_1$ in
     the event that the most recent mini-phase succeeded.  The full
     filling algorithm consists of $p - 1$ phases.

     We claim that each phase $i$ has constant probability of ending
     in a successful mini-phase (and thus swapping cup $i$ with a new
     cup $j \ge p$ having fill $\ge c_1$). Using this claim, one can
     complete the analysis as follows. If the swap in phase $i$ is at
     the end of a successful mini-phase, then after the swap, the
     (new) cup $i$ will have fill $\ge c_1$, and will continue to have
     fill $\ge c_1$ for the rest of the filling algorithm, since the
     filler puts $1$ unit in cup $i$ during every remaining step. At
     the end of the algorithm, the number of cups with fill $\ge c_1$
     is therefore bounded below by a sum of $p - 1$ independent $0$-$1$ random
     variables with total mean $\Omega(p)$. This means that the number
     of such cups with fill $\ge c_1$ is at least $\Omega(p)$ with
     probability $1 - e^{-\Omega(p)}$, as desired.

     It remains to analyze the probability that a given phase $i$ ends
     with a successful mini-phase. 

     Call a mini-phase \defn{clean} if the emptier removes $1$ unit of
     water from each cup $1, 2, \ldots, p - 1$ during each step of the
     mini-phase, and \defn{dirty} otherwise. Because each dirty
     mini-phase increases the total amount of water in cups $1, 2,
     \ldots, p - 1$ by at least $1$, and because the emptying
     algorithm prevents backlog from ever exceeding $f(n)$, there can
     be at most $O(pf(n))$ dirty mini-phases during phase $i$.

     By Lemma \ref{lem:fillingstrategy}, each mini-phase
     (independently) has at least a constant probability of either
     being dirty or of succeeding. Out of the $f(n)n^2$ possible
     mini-phases in phase $i$, there can only be $O(f(n)p) \le
     o(f(n)n^2)$ dirty mini-phases. It follows that, with probability
     $1 - e^{-\Omega(f(n)n^2)}$, at least a constant fraction of the
     possible mini-phases $s$ succeed (or would have succeeded in the
     event that $w_i$ were at least as large as $s$). Thus the
     $w_i$-th mini-phase succeeds with constant probability.
   \end{proof}

   Combining the preceding lemmas, we prove Theorem
   \ref{thm:lowerbound}.
   \begin{proof}[Proof of Theorem \ref{thm:lowerbound}]
     If $p \ge \Omega(\log n / \log \log n)$, then the theorem follows
     immediately from Lemma \ref{lem:lowerbound2}. On the other hand,
     if $p$ is a sufficiently large constant factor smaller than $\log
     n / \log \log n$, then the theorem follows from Lemma
     \ref{lem:lowerbound1}.
   \end{proof}

\section{Lower Bounds Against Unending Guarantees}
\label{sec:unending}


In this section, we prove upper and lower bounds for \defn{unending
  guarantees}, which are probabilistic guarantees that hold for each
step $t$, even when $t$ is arbitrarily large. As a convention, we will
use $f_j(t)$ to denote the fill of cup $j$ after step $t$.

The main result of the section is a lower bound showing that no
\defn{monotone stateless emptier} can achieve an unending guarantee of
$o(\log n)$ backlog.

\begin{definition}
  An emptier is said to be \defn{stateless} if the
emptier's decision depends only on the state of the cups at each
step. An emptier is said to be \defn{monotone} if the following holds:
given a state $S$ of the cups in which the emptier selects some cup
$j$ to empty, if we define $S'$ to be $S$ except that the amount of
water in some cup $i \neq j$ has been reduced, then the emptier still
selects cup $j$ in state $S'$. A \defn{monotone stateless emptier} is
any emptier that is both monotone and stateless.

The monotonicity and stateless property dictate only how the emptier
selects a cup $j$ in each step. Once a cup $j$ is selected the emptier
is permitted to either (a) remove $1$ full unit of water from that
cup, or (b) skip their turn. This decision is allowed to be an
arbitrary function of the state of the cups.
\end{definition}

We begin in Section \ref{sec:emptier-class} by showing that all
monotone stateless emptiers can be modeled as using a certain type of
\defn{score function} to make emptying decisions.

In Section \ref{sec:filling}, we give an oblivious filling strategy,
called the \defn{fuzzing algorithm}, that prevents monotone stateless
emptiers from achieving unending probabilistic guarantees of
$o(\log n)$ backlog (in fact, the filling strategy places an expected
$\Theta(n^{2/3} \log n)$ water into $\Theta(n^{2/3})$ cups, meaning
that bounds on tail size are also not viable, unless backlog is
allowed to be polynomially large). The fuzzing algorithm is named
after what is known as the \defn{fuzzing technique} \cite{fuzzing} for
detecting security vulnerabilities in computer systems -- by barraging
the system with random noise, one accidentally discovers and exploits
the structural holes of the system.

In Section \ref{sec:time} we show that the fuzzing algorithm continues
to prevent unending guarantees, even when the emptier is equipped with a
global clock, allowing for the emptier to adapt to the number of steps
that have occurred so far in the game.

Finally, in Sections \ref{sec:specific_epsilon1} and
\ref{sec:specific_epsilon2}, we determine the exact values of the
resource-augmentation parameter $\epsilon$ for which the smoothed
greedy and asymmetric smoothed greedy emptying algorithms achieve
single-processor unending guarantees. In particular, we show that the
minimum attainable value of $\epsilon$ is $2^{-\polylog n}$.

\subsection{Score-Based Emptiers}\label{sec:emptier-class}

In this section, we prove an equivalence between monotone stateless
emptiers, and what we call score-based emptiers. We then state several
useful properties of score-based emptiers.

A \defn{score-based emptier} has \defn{score functions}
$\sigma_1, \sigma_2, \ldots, \sigma_n$. When selecting which cup to
empty from, the emptier selects the cup $j$ whose fill $f_j$ maximizes
$\sigma_j(f_j)$. The emptier can then select whether to either (a)
remove $1$ full unit of water from the cup, or (b) skip their turn;
this decision is an arbitrary function of the state of the cups. The score
functions are required to be monotonically increasing functions,
meaning that $\sigma_i(a) < \sigma_i(b)$ whenever $a < b$. Moreover,
in order to break ties, all of the scores in the multiset
$\{\sigma_i(j / 2) \mid i \in [n], j \in \mathbb{Z}^+\}$ are required
to be distinct. (We only consider fills of the form $j / 2$ because in
our lower bound constructions all fills will be multiples of $1/2$.)

It is easy to see that any score-based emptier is also a monotone
stateless emptier. The following theorem establishes that the other
direction is true as well:
\begin{theorem}
  Consider cup games in which the filler always places water into cups
  in multiples of $1/2$. For these cup games, every monotone stateless
  emptying algorithm is equivalent to some score-based emptying
  algorithm.
  \label{thm:emptier-equivalence}
\end{theorem}

For a set of cups $1, 2, \ldots, k$, a \defn{state of the cups} is a
tuple $S = \langle S(1), S(2), \ldots, S(k)\rangle$, where $S(j)$
indicates the amount of water in cup $j$. Throughout this section we
will restrict ourselves to states where $S(j)$ is a non-negative
integer multiple of $1/2$.

In order to prove Theorem \ref{thm:emptier-equivalence}, we first
derive several natural properties of monotone stateless emptiers. We
say that the pair $(j_1, r_1)$ \defn{dominates} the pair $(j_2, r_2)$
if either (a) $j_1 = j_2$ and $r_1 > r_2$, or if (b) $j_1 \neq j_2$
and in the cup state where the only two non-empty cups are $j_1$ and
$j_2$ with $r_1$ and $r_2$ water, respectively, the emptier selects
cup $j_1$. We say that a cup $j_1$ \defn{dominates} a cup $j_2$ in a
state $S$ if $(j_1, r_1)$ dominates $(j_2, r_2)$, where $r_1$ and
$r_2$ are the amounts of water in cups $j_1$ and $j_2$, respectively,
in state $S$.

The next lemma shows that the emptiers decision in each step is
determined by which cup dominates the other cups.
\begin{lemma}
  Let $S$ be any state of the cups $1, 2, \ldots, n$, and suppose the
  emptier is following a monotone stateless algorithm. Then the cup
  $j$ that the emptier selects from $S$ is the unique cup $j$ that
  dominates all other cups.
  \label{lem:dominates}
\end{lemma}
\begin{proof}
  It suffices to show that cup $j$ dominates all other cups, since
  only one cup can have this property. Consider a cup $j' \neq j$, and
  let $r_1$ and $r_2$ be the amounts of water in cups $j$ and $j'$,
  respectively, in state $S$. Let $S'$ be the state in which the only
  non-empty cups are $j$ and $j'$ with $r_1$ and $r_2$ units of water,
  respectively. By the monotonicity property of the emptier, it must
  be that the emptier selects cup $j$ over cup $j'$ in state
  $S'$. Thus cup $j$ dominates cup $j'$ in state $S$, as desired.
\end{proof}

Next we show that domination is a transitive property.
\begin{lemma}
  Consider any monotone stateless emptying algorithm. If $(j_1, r_1)$
  dominates $(j_2, r_2)$ and $(j_2, r_2)$ dominates $(j_3, r_3)$, then
  $(j_1, r_1)$ dominates $(j_3, r_3)$.
  \label{lem:transitivity}
\end{lemma}
\begin{proof}
  We begin by considering the case where $j_1, j_2, j_3$ are
  distinct. Consider the cup state $S$ in which the only three cups
  that contain water are $j_1, j_2, j_3$, and they contain
  $r_1, r_2, r_3$ water, respectively. By Lemma \ref{lem:dominates},
  one of cups $j_1, j_2, j_3$ must dominate the others. Since $j_2$ is
  dominated by $j_1$ and $j_3$ is dominated by $j_2$, it must be that
  $j_1$ is the cup that dominates. Thus $(j_1, r_1)$ dominates
  $(j_3, r_3)$, as desired.

  Next we consider the case where $j_1 = j_2$ and $j_2 \neq
  j_3$. Suppose for contradiction that $(j_1, r_1)$ does not dominate
  $(j_3, r_3)$. Consider the cup state $S$ in which $j_1$ and $j_3$
  are the only cups containing water, and they contain $r_1$ and $r_3$
  units of water, respectively. In state $S$, cup $j_3$ dominates cup
  $j_1$. By monotonicity, it follows that if we decrease the fill of
  $j_2$ to $r_2$, then cup $j_3$ must still dominate cup $j_1 = j_2$. But
  this means that $(j_3, r_3)$ dominates $(j_2, r_2)$, a
  contradiction.

  Next we consider the case where $j_1 \neq j_2$ and $j_2 = j_3$.
  Consider the cup state $S$ in which $j_1$ and $j_2$ are the only
  cups containing water, and they contain $r_1$ and $r_2$ units of
  water, respectively. In state $S$, cup $j_1$ dominates cup $j_2$. By
  monotonicity, it follows that if we decrease the fill of $j_2$ to
  $r_3$, then cup $j_1$ must still dominate cup $j_2 = j_3$. This
  means that $(j_3, r_3)$ dominates $(j_2, r_2)$, as desired.

  Finally we consider the case where $j_1 = j_2 = j_3$. In this case
  it must be that $r_1 > r_2$ and $r_2 > r_3$. Thus $r_1 > r_3$,
  meaning that $(j_1, r_1)$ dominates $(j_3, r_3)$, as desired.
\end{proof}

By exploiting the transitivity of the domination property, we can now
prove Theorem \ref{thm:emptier-equivalence}.
\begin{proof}[Proof of Theorem \ref{thm:emptier-equivalence}]
  Consider the set
  $X = \{(j, r) \mid j \in [n], r \in \{0, 0.5, 1, 1.5, \ldots\}\}$. For \\
  $(j_1, r_1), (j_2, r_2) \in X$, say that $(j_1, r_1) < (j_2, r_2)$
  if $(j_1, r_1)$ dominates $(j_2, r_2)$. By Lemma
  \ref{lem:transitivity}, the set $X$ is totally ordered by the $<$
  operation. Since every totally ordered set is also well ordered, and
  since every countably-infinite well ordered set is order-isomorphic
  to the natural numbers \cite{settheory}, it follows that $X$ is
  order isomorphic to the natural numbers. That is, there is a
  bijection $\sigma:X \rightarrow \mathbb{N}$ that preserves the $<$
  relationship.

  Let $\sigma_i:\{0, 0.5, 1, 1.5, \ldots\} \rightarrow \mathbb{N}$ be
  the function $\sigma_i(r) = \sigma((i, r))$. By Lemma
  \ref{lem:dominates}, the emptier always selects the cup $j$ whose
  fill $f_j$ maximizes $\sigma_j(f_j)$. It follows that the emptier is
  a score-based emptier.
\end{proof}

We conclude the section by observing a useful property of score-based
emptiers, namely the existence of what we call equilibrium states.

We say that a state on $k$ cups is an \defn{equilibrium state} if for
every pair of distinct cups $i, j \in \{1, 2, \ldots, k\}$,
$\sigma_i(S(i) + 1/2) > \sigma_j(S(j))$. That is, for any cup $i$, if
$1/2$ unit of water is added to any cup $i$, then that cup's score
function will exceed the score function of all other cups
$\{1, 2, \ldots, k\} \setminus \{i\}$.

\begin{lemma}
  Consider cups $1, 2, \ldots, k$, and suppose their total fill $m$
  is a non-negative integer multiple of $1/2$. For any set of score
  functions, $\sigma_1, \ldots, \sigma_k$, there is a unique
  equilibrium state for cups $1, 2, \ldots, k$ in which the total
  amount of water in the cups is $m$.
  \label{lem:equilibrium}
\end{lemma}
\begin{proof}
  Consider any state $S = \langle S(1), \ldots, S(k) \rangle$ for cups
  $1, 2, \ldots, k$ in which the total fill of the cups is $m$. Define
  the \defn{score severity} of $S$ to be
  $\max_{i \in [k]} \sigma_i(S(i))$. If $S$ is not an equilibrium
  state, then we can move $1/2$ units of water from some cup
  $i \in \{1, 2, \ldots, k\}$ to some other cup
  $j \in \{1, 2, \ldots, k\}$ in a way that decreases the score
  severity of $S$.

  Let $\mathcal{S}_m$ be the set of states for cup $1, 2, \ldots, k$
  in which each cup contains a multiple of $1/2$ units of water, and
  in which the total amount of water in cups is $m$. Since
  $\mathcal{S}_m$ is finite, there must be a state
  $S \in \mathcal{S}_m$ with minimum score severity. By the preceding
  paragraph, it follows that $S$ is an equilibrium state.

  Finally, we prove uniqueness. Suppose $S, S' $ are distinct
  equilibrium states in $\mathcal{S}_m$. Then some cup $i$ in $S' $
  must have greater fill than the same cup $i$ in $S$. But by the
  equilibrium property, adding $1/2$ units of water to cup $i$ in $S$
  increases the score function of cup $i$ to be larger than any other
  cup's score function in $S$. Thus $S'$ must have a larger score
  severity than does $S$. Likewise, $S$ must have a larger score
  severity than $S'$, a contradiction.
\end{proof}

\subsection{The Oblivious Fuzzing Filling Algorithm}\label{sec:filling}

In this section, we describe a simple filling algorithm that, when
pitted against a score-based emptier, achieves backlog
$\Omega(\log n)$ after $n^{\Theta(n \log n)}$ steps with at least constant
probability. Note that, throughout this section, we focus only on cup
games that do \emph{not} have resource augmentation.

The filling strategy, which we call the \defn{oblivious fuzzing
  algorithm} has a very simple structure. At the beginning of the
algorithm, the filler randomly permutes the labels $1, 2, \ldots, n$
of the cups. The filler then begins their strategy by spending a large
number (i.e., $n^{\Theta(n \log n)}$) of steps randomly placing water
into cups $1, 2, \ldots, n$. The filler then disregards cup $n$ (note
that cup $n$ is a random cup due to the random-permutation step!), and
spends a large number of steps randomly placing water into cups
$1, 2, \ldots, n - 1$. The filler then disregards cup $n - 1$ and
spends a large number of steps randomly placing water into cups
$1, 2, \ldots, n - 2$, and so on.

Formally, the oblivious fuzzing algorithm works as follows. Let
$c \in \mathbb{N}$ be a sufficiently large constant, and relabel the
cups (from the filler's perspective) with a random permutation of
$1, 2, \ldots, n$. The filling strategy consists of $n$ phases of
$n^{cn}$ steps. The $i$-th phase is called the \defn{$(n - i + 1)$-cup
  phase} because it focuses on cups $1, 2, \ldots, (n - i + 1)$. In
each step of the $i$-th phase, the filler selects random values
$x_1, x_2 \in \{1, 2, \ldots, n - i + 1\}$ uniformly and
independently, and then places $\frac{1}{2}$ water into each of cups
$x_1, x_2$. If $x_1 = x_2$, then the cup $x_1$ receives a full unit
of water.

One interesting characteristic of the oblivious fuzzing algorithm
is that it represents a natural workload in the scheduling problem
that the cup game models. One can think of the cups as representing
$n$ tasks and water as representing work that needs to be
scheduled. In this scheduling problem, the oblivious fuzzing
filling algorithm simply assigns work to tasks at random, and selects
one task every $n^{cn \log n}$ steps to stop receiving new work.

In this section, we prove the following theorem.

\begin{theorem}
  Consider a cup game on $n$ cups. Suppose that the emptier follows a
  score-based emptying algorithm, and that the filler follows the
  oblivious fuzzing filling algorithm. Then at the beginning of
  the $n^{2/3}$-cup phase, the average fill of cups
  $1, 2, \ldots, n^{2/3}$ is $\Omega(\log n)$, in expectation.
\label{thm:general-filler-strategy}
\end{theorem}

For each $\ell \in \{1, 2, \ldots, n - 1\}$, call a step $t$ in the
$\ell$-cup phase \defn{emptier-wasted} if the emptier fails to remove
water from any of cups $1, 2, \ldots \ell$ during step $t$ (either
because the emptier skips their turn, or because the emptier selects a
cup $j > \ell$). We show that for each
$\ell \in \{n^{2/3} + 1, n^{2/3} + 2, \ldots, n - 1\}$, the $\ell$-cup
phase has at least $\Omega(1)$ emptier-wasted steps in expectation (or
the average height of cups in that phase is already $\Omega(\log
n)$). During an emptier-wasted step $t$, the total amount of water in
cups $1, 2, \ldots, \ell$ increases by $1$ (since the filler places
water into the $\ell$ cups, and the emptier does not remove water from
them). It follow that, during the $\ell$-cup phase, the \emph{average}
amount of water in cups $1, 2, \ldots, \ell$ increases by
$\Omega\left(\frac{1}{\ell}\right)$ in expectation. Applying this
logic to every phase gives Theorem
\ref{thm:general-filler-strategy}. The key challenge is show that,
within the $\ell$-cup phase, the expected number of emptier-wasted
steps is $\Omega(1)$.

For each $\ell \in \{1, 2, \ldots, n - 1\}$, define the \defn{initial
  water level $m_\ell$ of the $\ell$-cup phase} to be the total amount
of water in cups $1, 2, \ldots, \ell + 1$ at the beginning of the
phase. Define the \defn{equilibrium state
  $E_\ell = \langle E_\ell(1), \ldots, E_\ell(\ell + 1) \rangle$ for
  the $\ell$-cup phase} to be the equilibrium state for cups
$1, 2, \ldots, \ell + 1$ in which the total amount of water is
$m_\ell + 1$ (note that $E_\ell$ exists and is unique by Lemma
\ref{lem:equilibrium}). One can think of $m_\ell + 1$ as representing
the total amount of water in cups $1, 2, \ldots, \ell + 1$ after the
filler places $1$ unit of water into the cups at the beginning of the
first step in the $\ell$-cup phase.

Define the \defn{bolus $b_\ell$ of the $\ell$-cup phase} as
follows. If $r$ is the amount of water in cup $\ell + 1$ at the
beginning of the $\ell$-cup phase, then
$b_\ell = \max(0, r - E_\ell(\ell + 1))$. That is, $b_\ell$ is the
amount by which cup $\ell + 1$ exceeds its equilibrium fill.

We begin by showing that, if $m_\ell \le O(n \log n)$, then the
expected number of emptier-wasted steps in phase $\ell$ is at least
$\E[b_\ell / 2]$. The basic idea is that, whenever fewer than $b_\ell$
emptier-wasted steps have occurred, the filler has some small
probability of reaching a state in which all of cups
$1, 2, \ldots, \ell$ have fills no greater than
$E_\ell(1), E_\ell(2), \ldots, E_\ell(\ell)$, respectively. If this
happens, then the score function of cup $\ell + 1$ will exceed that of
any of cups $1, 2, \ldots, \ell$, and an emptier-wasted step
occurs. Thus, whenever fewer than $b_\ell$ emptier-wasted steps have
occurred, the filler has a small probability of incurring an
emptier-wasted step (within the next $O(n \log n)$ steps). Since the
$\ell$-cup phase is very long, the filler has many opportunities to
induce an emptier-wasted step in this way. It follows that, with high
probability, there will be at least $b_\ell$ emptier-wasted steps in
the $\ell$-cup phase. Lemma \ref{lem:bolus-to-wasted-steps} presents
this argument in detail.

\begin{lemma}
  Let $\ell \in [n - 1]$, condition on $m_\ell \le n \log n$, and
  condition on some value for $b_\ell$. Under these conditions, the
  the expected number of emptier-wasted steps in the $\ell$-cup phase
  is at least $b_\ell / 2$.
  \label{lem:bolus-to-wasted-steps}
\end{lemma}
\begin{proof}
  Call a step in the $\ell$-cup phase \defn{equilibrium-converging} if
  for each cup $j$ that the filler places water into, the fill $x$ of
  cup $j$ after the water is placed satisfies $x \le E_\ell(j)$. One
  can think of an equilibrium-converging step as being a step in which
  the filler's behavior pushes each cup $j$ towards its equilibrium
  state, without pushing any cups above their equilibrium state.

  Call a step in the $\ell$-cup phase a \defn{convergence-enabling} if
  the total amount of water in cups $1, 2, \ldots, \ell$ is less than
  $\left( \sum_{i = 1}^\ell E_\ell(i) \right) - 1$ at the beginning of
  the step.
  
  Convergence-enabling steps have two important
  properties. \\
  \textbf{The Convergence Property: }For any convergence-enabling step
  $t$, there is some pair of cups $j, k$ (possibly $j = k$) that the
  filler can place water into in order so that the step is equilibrium
  converging. Thus, whenever a convergence-enabling step occurs, there
  is probability of at least $1/\ell^2$ that the step is equilibrium
  converging. \\
  \textbf{The Bolus Property: }At the beginning of any
  convergence-enabling step, the amount of water in cup $\ell + 1$
  must be greater than $E_{\ell}(\ell + 1)$. This is a consequence of
  the fact that the total amount of water in cups
  $1, \ldots, \ell + 1$ is at least
  $m_\ell = \left(\sum_{i = 1}^{\ell + 1} E_{\ell}(i)\right) - 1$.

  \vspace{.2 cm}

  Break the $\ell$-cup phase into sequences of steps
  $L_1, L_2, L_3, \ldots$, where each $L_i$ is $2n\log n$
  steps. We begin by showing that, if $L_i$ contains a
  convergence-enabling step and consists of only
  equilibrium-converging steps, then $L_i$ must also contain at least
  one emptier-wasted step.
  \begin{claim}
    Suppose $m_\ell \le n \log n$. Suppose that the first step of
    $L_i$ is convergence-enabling. If all of the steps in $L_i$ are
    equilibrium converging, then at least one of the steps must be
    emptier-wasted.
    \label{clm:convergence-to-wasted}
  \end{claim}
  \begin{proof}    
    At the end of each step $t$, let $f_j(t)$ denote the amount of
    water in each cup $j$. Define the potential function $\phi(t)$ to
    be
    $$\phi(t) = \sum_{j = 1}^\ell \begin{cases}1 + f_j(t) - E_\ell(j)  & \text{ if } f_j(t) > E_\ell(j), \\  0 & \text{ otherwise.} \end{cases}$$
    Since the first step of $L_i$ is convergence-enabling, the total
    amount of water in the cups at the beginning of $L_i$ is at most
    $\left(\sum_{i - 1}^\ell E_\ell(i)\right) - 1 \le m_\ell \le n
    \log n$. It follows that, at the beginning of $L_i$, the potential
    function $\phi$ is at most $n \log n + n \le 2n \log n - 1$.

    Whenever a step $t$ is both equilibrium-converging and
    non-emptier-wasted, we have that either $\phi(t - 1) = 0$ or
    $\phi(t) < \phi(t - 1) - 1$. Since $\phi$ is at most
    $2n \log n - 1$ at the beginning of $L_i$, we cannot have
    $\phi(t) < \phi(t - 1) - 1$ for every step in $L_i$. Thus, if
    every step in $L_i$ is equilibrium converging, then there must be
    at least one step that is either emptier-wasted or that satisfies
    $\phi(t - 1) = 0$.

    To complete the claim, we show that if there is at least one step
    in $L_i$ for which $\phi(t - 1) = 0$, and step $t$ is
    equilibrium-converging, then there also be at least one
    emptier-wasted step. Suppose $\phi(t - 1) = 0$, that step $t$ is
    equilibrium-converging, and that no steps in $L_i$ are
    emptier-wasted. Since there are no emptier-wasted steps in $L_i$,
    every step in $L_i$ must be equilibrium-enabling, and thus cup
    $\ell + 1$ contains more than $E_{\ell}(\ell + 1)$ water at the
    beginning of step $t$ (by the Bolus Property of
    equilibrium-enabling steps). Since $\phi(t - 1) = 0$ and step $t$
    is equilibrium converging, the cups $1, 2, \ldots, \ell$ contain
    fills at most $E_\ell(1), E_\ell(2), \ldots, E_\ell(\ell)$,
    respectively, after the filler places water in step $t$. It
    follows that, during step $t$, the emptier will choose cup
    $\ell + 1$ over all of cups $1, 2, \ldots, \ell$. Thus step $t$ is
    an emptier-wasted step, a contradiction.
  \end{proof}

  Next we use Claim \ref{clm:convergence-to-wasted} in order to show
  that, if $m_\ell \le n \log n$ and $L_i$ contains a
  convergence-enabling step, then $L_i$ has probability at least
  $1/n^{4n\log n}$ of containing an emptier-wasted step.
  \begin{claim}
    Condition on the fact that the first step of $L_i$ is convergence-enabling and
    that $m_\ell \le n \log n$.  Then $L_i$ contains an emptier-wasted
    step with probability at least $1/n^{4n \log n}$.
    \label{clm:convergence-to-wastedb}
  \end{claim}
  \begin{proof}
    Since the first step of $L_i$ is convergence-enabling, either
    every step of $L_i$ is convergence-enabling or there is at least
    one emptier-wasted step. Recall by the Convergence Property that
    each convergence-enabling step has probability at least $1/n^2$ of
    being equilibrium-converging. Thus there is probability at least
    $1/n^{4n\log n}$ that every step of $L_i$ (up until the first
    emptier-wasted step) is equilibrium-converging. By Claim
    \ref{clm:convergence-to-wasted}, it follows that the probability
    of there being an emptier-wasted step is at last $1/n^{4n\log n}$.
  \end{proof}

  We can now complete the proof of the lemma. For each $L_i$, if the
  number of emptier-wasted steps in $L_1, \ldots, L_{i - 1}$ is less
  than $\lceil b_\ell \rceil$, then the first step of $L_i$ is
  convergence-enabling. Since $m_\ell \le n \log n$,
  then by Claim \ref{clm:convergence-to-wastedb}, it follows that $L_i$
  has probability at least $1/n^{4n}$ of containing an emptier-wasted
  step.

  Now collect the $L_i$'s into collections of size $n^{4n \log n + 1}$, so
  that the $k$-th collection is given by
  $$C_k = \langle L_{(k - 1) n^{4n\log n + 1} + 1}, \ldots, L_{k n^{4n\log n + 1}} \rangle.$$ Note
  that, as long as the constant $c$ used to define the fuzzing
  algorithm is sufficiently large, then the $l$-cup phase is long
  enough so that it contains at least
  $ \lceil b_\ell \rceil \le m_\ell \le n \log n$ collections
  $C_1, C_2, \ldots, C_{\lceil b_\ell \rceil}$.

  Say that a step collection $C_i$ \defn{failed} if, at the beginning
  of the step collection, the number of emptier-wasted steps that have
  occurred is less than $b_\ell$, and $C_i$ contains no emptier-wasted
  steps. The probability of a given $C_i$ failing is at most,
  $$\left(1 - 1 / n^{4n \log n} \right)^{n^{4n \log n + 1}}   \le 1/e^n.$$

  It follows that the probability of any of
  $C_1, \ldots, C_{\lceil b_\ell \rceil}$ failing is at most, 
  $$1 - (1 - 1/e^n)^{\lceil b_\ell \rceil} \le 1 -  (1 - 1/e^n)^{n \log n} \le 1 / 2.$$
  
  If none of the collections $C_1, \ldots, C_{\lceil b_\ell \rceil}$
  fail, then there must be at least $\lceil b_\ell \rceil$
  emptier-wasted steps. Thus the expected number of emptier-wasted
  steps that occur during the phase is at least $b_\ell / 2$.
\end{proof}

In order to show that the expected number of emptier-wasted steps in
phase $\ell$ is $\Omega(1)$ (at least, whenever
$m_\ell \le n \log n$), it suffices to show that expected bolus
$b_\ell$ is $\Omega(1)$ (conditioned on $m_\ell \le n \log n$).

In order to prove a lower-bound on the bolus, we examine a related
quantity that we call the \defn{variation}. If $t + 1$ is the first step
of the $\ell$-cup phase, then the \defn{variation $v_\ell$} of the $\ell$-cup
phase is defined to be, 
$$v_\ell = \sum_{j = 1}^{\ell + 1} |f_j(t) - E_\ell(j)|.$$
The variation $v_\ell$ captures the degree to which the fills of cups
$1, 2, \ldots, \ell + 1$ differ from their equilibrium fills. The next
lemma shows that, if the variation $v_\ell$ is large, then so will be
the bolus $b_\ell$ in expectation.
\begin{lemma}
  Let $\ell \in \{1, 2, \ldots, n - 1\}$. Fix some value of $v_\ell$
  and of $m_\ell$. Then
  $$\E[b_\ell] = \frac{v_\ell}{2(\ell + 1)}.$$
  \label{lem:variation-to-bolus}
\end{lemma}
\begin{proof}
  Let $t + 1$ be the first step in the $\ell$-cup phase. By the definition of
  $E_\ell$, we have that,
  $$\sum_{j = 1}^{\ell + 1} f_j(t) = \sum_{j = 1}^{\ell + 1} E_\ell(j).$$
  Thus,
  $$\sum_{j = 1}^{\ell + 1} \max(0, f_j(t) - E_\ell(j)) = \sum_{j = 1}^{\ell + 1} \max(0, E_{\ell}(j) - f_j(t)).$$
  Hence,
  $$\sum_{j = 1}^{\ell + 1} \max(0, f_j(t) - E_\ell(j)) = v_\ell / 2.$$
  Since the cups $1, 2, \ldots, \ell + 1$ are randomly labeled, we have by symmetry that,
  $$\E[\max(0, f_{\ell + 1}(t) - E_\ell(\ell + 1))] = \frac{1}{\ell + 1} \E\left[\sum_{j = 1}^{\ell + 1} \max(0, f_j(t) - E_\ell(j))\right] = \frac{v_\ell}{2(\ell + 1)}.$$
  Since $b_\ell = \max(0, f_{\ell + 1}(t) - E_\ell(\ell + 1))$, the proof of the lemma is complete.
\end{proof}

By Lemma \ref{lem:variation-to-bolus}, if our goal is to show that
$\E[b_\ell \mid m_\ell \le n \log n] \ge \Omega(1)$, then it suffices to show that
$\E[v_\ell \mid m_\ell \le n \log n] \ge \Omega(\ell)$. 
\begin{lemma}
  Let $n > 1$.  For $\ell \in \{n^{2/3} + 1, \ldots, n - 1\}$, the
  variation $v_\ell$ satisfies,
  $\E[v_\ell \mid m_\ell \le n \log n] \ge \Omega(\ell)$.
  \label{lem:variation}
\end{lemma}
\begin{proof}
  Let $t + 1$ be the first step of the $\ell$-cup phase. Recall that
  $$v_\ell = \sum_{j = 1}^{\ell + 1} |f_j(t) - E_\ell(j)|.$$
  Note that the equilibrium state $E_{\ell}$ depends on the amount of
  water $m_\ell$ in cups $1, 2, \ldots, \ell + 1$ at the beginning of
  step $t$. Let $E^{(m)}_\ell$ denote the equilibrium state in the
  case where $m_\ell = m$, and let $v_\ell^{(m)}$ denote the variation of
  the cups at step $t$ from $E^{(m)}_\ell$, i.e.,
  $$v_\ell^{(m)} = \sum_{j = 1}^{\ell + 1} |f_j(t) - E^{(m)}_\ell(j)|.$$
  Since we are conditioning on $m_\ell \le n \log n$, it suffices to
  consider values of
  $m \in \{k/2 \mid k \in \{0, 1, 2, \ldots, 2n\log n\}\}$. For each
  such value of $m$, we will show that
  $\Pr[v_\ell^{(m)} = \Omega(\ell)] \ge 1 - O(1/n^2)$. By a union
  bound, it follows that,
  $\Pr[v_\ell = \Omega(\ell) \mid m_\ell \le n \log n] \ge 1 - O(\log n / n)$. It follows that
  $\E[v_\ell \mid m_\ell \le n \log n] \ge \Omega(\ell)$, as desired.
  
  To complete the proof of the lemma, we must examine
  $\Pr[v_\ell^{(m)} = \Omega(\ell)]$. To do this, we break the water
  placed by the filler into two parts: Let $a_j$ denote the amount of
  water placed into each cup $j$ by the filler in the first $n$ steps
  of the first phase, and let $b_j$ denote the amount of water placed
  into each cup $j$ by the filler in steps
  $n + 1, n + 2, \ldots, t - 1$. Finally, let $c_j$ denote the total
  amount of water removed from cup $j$ by the emptier during steps
  $1, 2, \ldots, t - 1$.

  The role of $a_j$ will be similar to that of the random offsets in
  the smoothed greedy emptying algorithm. Interestingly, these random
  offsets now work in the \emph{filler's favor}, rather than the
  emptier's. 

  Consider the quantity $X_j = f_j(t) - E^{(m)}_\ell(j) \pmod 1$. We
  can lower bound the variation $v_\ell^{(m)}$ by,
\begin{equation}
v_\ell^{(m)} \ge \sum_{j = 1}^{\ell + 1} X_j. 
\label{eq:v-to-X}
\end{equation}
Each $X_j$ can be written as, 
$$X_j = a_j + b_j - c_j - E^{(m)}_\ell(j) \pmod 1.$$
Since the emptier always removes water in chunks of size $1$,
$c_j \equiv 0 \mod 1$. Thus,
\begin{equation}
X_j = a_j + b_j - E^{(m)}_\ell(j) \pmod 1.
\label{eq:aj-role}
\end{equation}
For the sake of analysis, fix the filler's behavior in steps
$n + 1, n + 2, \ldots$, meaning that the only randomness is in the
$a_j$'s and the $b_j$'s are fixed. Define
$d_j = b_j - E^{(m)}_\ell(j)$. By \eqref{eq:v-to-X} and
\eqref{eq:aj-role}, our goal is to show that
\begin{equation}
\sum_{j = 1}^{\ell} 
\begin{cases} 
1/2 \text{ if } a_j \not\equiv d_j \mod 1 \\
0 \text{ otherwise}
\end{cases}
\ge \Omega(\ell),
\label{eq:aj-goal}
\end{equation}
with probability at least $1 - O(1/n^2)$. 

We begin by showing that the left side of \eqref{eq:aj-goal} has
expected value $\Omega(\ell)$. Note that,
$\Pr[a_j = 0] = (1 - 1/n)^{2n} \ge 1/16$ and
$\Pr[a_j = 1/2] = \binom{2n}{1} \cdot 1/n \cdot (1 - 1/n)^{2n - 1} =  2(1 - 1/n)^{2n - 1} \ge 1/4$. Thus
$\Pr[a_j \not\equiv d_j \mod 1] \ge 1/16$, implying that the left
side of \eqref{eq:aj-goal} has expected value at least $\ell / 16$.

In order to prove that \eqref{eq:aj-goal} holds with probability at
least $1 - O(1/n^2)$, we show that the left side of \eqref{eq:aj-goal}
is tightly concentrated around its mean. If the $X_j$'s were
independent of one another then we could achieve this with a Chernoff
bound. Since the $X_j$'s are dependent, we will instead use
McDiarmid's inequality.
\begin{theorem}[McDiarmid '89 \cite{McDiarmid89}]
Let $A_1, \ldots, A_m$ be independent random variables over an arbitrary probability space. Let $F$ be a function mapping $(A_1, \ldots, A_m)$ to $\mathbb{R}$, and suppose $F$ satisfies,
$$\sup_{a_1, a_2, \ldots, a_n, \overline{a_i}} |F(a_1, a_2, \ldots, a_{i - 1}, a_i, a_{i + 1}, \ldots , a_n) - F(a_1, a_2, \ldots, a_{i - 1}, \overline{a_i}, a_{i + 1}, \ldots , a_n)| \le R,$$
for all $1 \le i \le n$. That is, if
$A_1, A_2, \ldots, A_{i - 1}, A_{i + 1}, \ldots, A_n$ are fixed, then
the value of $A_i$ can affect the value of $F(A_1, \ldots, A_n)$ by at
most $R$; this is known as the \defn{Lipschitz condition}. Then for
all $S > 0$,
$$\Pr[|F(A_1, \ldots, A_n) - \E[F(A_1, \ldots, A_n)]| \ge R \cdot S] \le 2e^{-2S^2 / n}.$$
\end{theorem}
We will apply McDiarmid's inequality to the quantity
$F = \sum_{i = 1}^{\ell + 1} X_i$ (i.e., the left side of
\eqref{eq:aj-goal}). Recall that the filler's behavior in steps
$n + 1, \ldots, t - 1$ has been fixed, meaning that the value of $F$
is a function of the filler's behavior in steps $1, 2, \ldots,
n$. Thus $F$ is a function of $2n$ independent random variables
$A_1, A_2, \ldots, A_{2n}$, where $A_{2i - 1}$ and $A_{2i}$ are the
cups that the filler places water into during step $i$. Moreover, the
outcome of each $A_i$ can only change the value of
$F(A_1, \ldots, A_{2n})$ by at most $\pm 1/2$. Thus
$F(A_1, \ldots, A_{2n})$ satisfies the Lipschitz condition with
$R = 1/2$.

We apply McDiarmid's inequality to
$F(A_1, \ldots, A_{2n}) = \sum_{i = 1}^{\ell + 1} X_i$, with $R = 1/2$
and $S = \ell / 16$ in order to conclude that,
$$\Pr[\sum_{i = 1}^{\ell + 1} X_i < \ell / 32] \le 2e^{\Omega(\ell^2 / n)} \le e^{\Omega(n^{1/3})} \le O(1/n^2).$$
Thus \eqref{eq:aj-goal} holds with probability at least $1 - O(1/n^2)$, as desired.
\end{proof}

Combining the preceding lemmas, we get that the expected number of
emptier-wasted steps in each phase is $\Omega(1)$.
\begin{lemma}
  Suppose $n > 1$. For $\ell \in \{n^{2/3} + 1, \ldots, n - 1\}$, if
  $m_\ell \le n \log n$, then the expected number of emptier-wasted
  steps in the $\ell$-cup phase is $\Omega(1)$.
  \label{lem:emptier-wasted2}
\end{lemma}
\begin{proof}
  By Lemma \ref{lem:bolus-to-wasted-steps}, the expected number of
  emptier-wasted steps is at least
  $\E[b_\ell / 2 \mid m_\ell \le n \log n]$. By Lemma
  \ref{lem:variation-to-bolus},
  $\E[b_\ell \mid m_\ell \le n \log n] = \E[v_{\ell} / 2(\ell + 1)
  \mid m_\ell \le n \log n]$. By Lemma \ref{lem:variation},
  $\E[v_{\ell} / 2(\ell + 1) \mid m_\ell \le n \log n] \ge
  \Omega(1)$. Thus the expected number of emptier-wasted steps in the
  $\ell$-cup phase is $\Omega(1)$.
\end{proof}

We can now complete the proof of Theorem \ref{thm:general-filler-strategy}.
\begin{proof}[Proof of Theorem \ref{thm:general-filler-strategy}]
  For each $\ell \in \{n^{2/3} + 1, \ldots, n\}$, let $\phi_\ell$
  denote the average fill of cups $1, 2, \ldots, \ell$ at the end of
  the $\ell$-cup phase. We will show that
  $\E[\phi_{n^{2/3} + 1}] = \Omega(\log n)$, completing the proof of
  the theorem.

  Consider the $\ell$-cup phase, where
  $\ell \in \{n^{2/3} + 1, \ldots, n\}$. At the beginning of the
  phase, the average amount of water $\psi_\ell$ in cups
  $1, 2, \ldots, \ell$ has expected value
  $\E[\psi_\ell] = \E[\phi_{\ell + 1}]$, since the cups
  $1, 2, \ldots, \ell + 1$ are indistinguishable up until the
  beginning of the $\ell$-cup phase. If $w_\ell$ is the expected
  number of emptier-wasted steps in phase $\ell$, then,
  $\E[\phi_\ell] = \E[\psi_\ell] + w_\ell / \ell$. Hence,
  $$\E[\phi_{n^{2/3}}] \ge \Omega\left(\frac{w_n}{n - 1} + \frac{n_{n - 1}}{n - 2} +
    \cdots + \frac{w_{n^{2/3} + 1}}{n^{2/3}}\right).$$

  Let $X$ denote the event that $m_\ell \le n \log n$ for all
  $\ell \in \{n^{2/3} + 1, \ldots, n\}$. If $\Pr[X] \ge 1/2$, then it
  follows from Lemma \ref{lem:emptier-wasted2} that
  $w_\ell \ge \Omega(1)$ for each
  $\ell \in \{n^{2/3} + 1, \ldots, n\}$. Thus,
  $$\E[\phi_{n^{2/3}}] \ge \Omega\left(\frac{1}{n} + \frac{1}{n - 2} +
    \cdots + \frac{1}{n^{2/3} + 1}\right) \ge \Omega(\log n),$$ as
  desired.

  Now consider the case where $\Pr[X] < 1/2$. That is, with
  probability at least $1/2$, there is some $\ell$ for which
  $m_\ell \ge n \log n$. Let $\ell$ be the largest $\ell$ for which
  $m_\ell \ge n \log n$. Then the average fill $\phi_\ell$ of cups
  $1, 2, \ldots, \ell + 1$ at the beginning of the $\ell$-cup phase is
  at least $\log n$. Note that, for any phase $r$ and value $k$, we
  have $\E[\phi_{r - 1} \mid \phi_r = k] \ge k$. Given that
  $\phi_\ell \ge \log n$, it follows that
  $\E[\phi_{\ell - 1}], \E[\phi_{\ell - 2}], \ldots \ge \log n$. This
  means that $\E[\phi_{n^{2/3}} \mid X] \ge \log n$. Since $X$ occurs
  with probability at least $1/2$, we have that
  $\E[\phi_{n^{2/3}}] \ge \Omega(\log n)$, completing the proof.
\end{proof}

\subsection{Giving the Emptier a Time Stamp}\label{sec:time}

In this section, we show that unending guarantees continue to be
impossible, even if the score-based emptier is permitted to change
their algorithm based on a global time stamp.

A \defn{dynamic score-based emptying algorithm} $\mathcal{A}$ is
dictated by a sequence
$\langle \mathcal{X}_1, \mathcal{X}_2, \mathcal{X}_3, \ldots \rangle$,
where each $X_i$ is a score-based emptying algorithm. On step $t$ of
the cup game, the algorithm $\mathcal{A}$ follows algorithm
$\mathcal{X}_t$.


Define the \defn{extended oblivious fuzzing filling algorithm} to
be the oblivious fuzzing filling strategy, except that each
phase's length is increased to consist of $T(n)$ steps, where $T$ is a
sufficiently large function of $n$ (that we will choose later).

\begin{theorem}
  Consider a cup game on $n$ cups. Suppose the emptier is a dynamic
  score-based emptier. Suppose the filler follows the extended
  oblivious fuzzing filling algorithm. Then at the beginning of
  the $n^{2/3}$-cup phase, the average fill of cups
  $1, 2, \ldots, n^{2/3}$ is $\Omega(\log n)$, in expectation.
\label{thm:timestamp}
\end{theorem}

\paragraph{Understanding when two score-based algorithms can be treated as ``equivalent''}
We say that the cups $1, 2, \ldots, n$ are in a \defn{legal state} if
each cup contains an integer multiple of $1/2$ water, and the total
water in the cups is at most $n \log n$. By the assumption that the
emptier is backlog-bounded, and that the filler follows the extended
oblivious fuzzing filling algorithm, we know that the cup game
considered in Theorem \ref{thm:timestamp} will always be in a legal
state.

Let $\mathbf{L}$ denote the set of legal states. For each score-based
emptying algorithm $\mathcal{X}$, we define the \defn{behavior vector}
$B(\mathcal{X})$ of $\mathcal{X}$ to be the set
$$B(\mathcal{X}) = \{(L, k) \mid \mathcal{X} \text{ empties from cup }k\text{ when the cups are in state }L \in \mathbf{L}\}.$$
Note that, for some states $L \in \mathbf{L}$, the emptier
$\mathcal{X}$ may choose not to empty from any cup---in this case, $L$
does not appear in any pair in $B(\mathcal{X})$.

The behavior vector $B(\mathcal{X})$ captures $\mathcal{X}$'s behavior
on all legal states. If $B(\mathcal{X}) = B(\mathcal{X}')$ for two
score-based emptying algorithms $\mathcal{X}$ and $\mathcal{X}'$, then
we treat the two emptying algorithms as being the same (since their
behavior is indistinguishable on the cup games that we are
analyzing). This means that the number of distinct score-based
emptying algorithms is finite, bounded by $(n +
1)^{|\mathcal{L}|}$. We will use $\mathcal{A}$ to denote the set of
distinct score-based emptying algorithms. Formally, each element of
$\mathcal{A}$ is an equivalence class of algorithms, where each
score-based algorithm is assigned to an equivalence class based on its
behavior vector $B(\mathcal{X})$.

\paragraph{Associating each phase with a score-based algorithm that it ``focuses'' on}
In order to analyze the $\ell$-cup phase of the extended oblivious
fuzzing filling algorithm, we break the phase into
\defn{segments}, where each segment consists of $2 n\log n \cdot |\mathcal{A}|$
steps. For each segment, there must be an algorithm $A \in \mathcal{A}$
that the emptier uses at least $2n\log n$ times within the segment. We say
that the segment \defn{focuses} on emptying algorithm $A$.

Let $K$ denote the number of segments in the $\ell$-cup phase. By the
pigeon-hole principle, there must be some algorithm $A \in \mathcal{A}$
such that at least $K / |\mathcal{A}|$ of the segments in the
$\ell$-cup phase focus on $A$. We say that the $\ell$-cup phase
\defn{focuses} on algorithm $A$.

For each $\ell \in \{1, \ldots, n\}$, let $A_\ell$ denote the
score-based emptying algorithm that the $\ell$-cup phase focuses on
(if there are multiple such algorithms $A_\ell$, select one
arbitrarily).

Our proof of Theorem \ref{thm:timestamp} will analyze the $\ell$-cup
phase by focusing on how the phase interacts with algorithm
$A_\ell$. This is made difficult by the fact that, between every two
steps in which the emptier uses algorithm $A_\ell$, there may be many
steps in which the emptier uses other score-based emptying algorithms.

\paragraph{Defining the equilibrium state and bolus of each phase}
We define the equilibrium state $E_\ell$ and the bolus $b_\ell$ of the
$\ell$-cup phase to each be with respect to the score-based emptying
algorithm $A_\ell$. That is, $E_\ell$ is the equilibrium state for
algorithm $A_\ell$ for the cups $1, 2, \ldots, \ell + 1$ in which the
total amount of water (in those cups) is $m_\ell + 1$ (recall that
$m_\ell$ is the amount of water in cups $1, 2, \ldots, \ell + 1$ at
the beginning of the $\ell$-cup phase). Using this definition of
$E_\ell$, the bolus is $b_\ell = \max(0, r - E_\ell(\ell + 1))$, where
$r$ is the amount of water in cup $\ell + 1$ at the beginning of the
$\ell$-cup phase.

The key to proving Theorem \ref{thm:timestamp} is to show that, if
$m_\ell \le n \log n$, then the expected number of emptier-wasted
steps in the $\ell$-cup phase is at least $\Omega(b_\ell)$. That is,
we wish to prove a result analogous to Lemma
\ref{lem:bolus-to-wasted-steps} from Section \ref{sec:filling}.

\begin{lemma}
  Let $\ell \in [n - 1]$, condition on $m_\ell \le n \log n$, and
  condition on some value for $b_\ell$. Under these conditions, the
  the expected number of emptier-wasted steps in the $\ell$-cup phase
  is at least $b_\ell / 2$.
  \label{lem:bolus-to-wasted-steps2}
\end{lemma}
\begin{proof}
  Call a step $t$ in the $\ell$-cup phase \defn{equilibrium-converging} if
  either:
  \begin{itemize}
  \item The emptier uses algorithm $A_\ell$ during step $t$, and for
    each cup $j$ that the filler places water into, the fill $x$ of
    cup $j$ after the water is placed satisfies $x \le E_\ell(j)$.
  \item The emptier uses an algorithm $A \neq A_\ell$ during step $t$,
    and the filler places all of their water (i.e., a full unit) into
    the cup $j$ whose score (as assigned by the score-based algorithm
    $A$) is largest at the beginning of step $t$.
  \end{itemize}
  The first case in the definition of equilibrium-converging steps is
  similar to that in the proof of Lemma
  \ref{lem:bolus-to-wasted-steps}. The second case, where the emptier
  uses an algorithm $A \neq A_\ell$ is different; in this case, the
  definition guarantees that the step is either emptier-wasted or is a
  \defn{no-op} (meaning that the water removed by the emptier during
  the step is exactly the same as the water placed by the filler).

  Call a step in the $\ell$-cup phase a \defn{convergence-enabling} if
  the total amount of water in cups $1, 2, \ldots, \ell$ is less than
  $\left( \sum_{i = 1}^\ell E_\ell(i) \right) - 1$ at the beginning of
  the step.

  Just as in the proof of Lemma \ref{lem:bolus-to-wasted-steps}, convergence-enabling steps have two important properties: \\
  \textbf{The Convergence Property: }For any convergence-enabling step
  $t$, there is some pair of cups $j, k$ (possibly $j = k$) that the
  filler can place water into in order so that the step is equilibrium
  converging. Thus, whenever a convergence-enabling step occurs, there
  is probability of at least $1/\ell^2$ that the step is equilibrium
  converging. \\
  \textbf{The Bolus Property: }At the beginning of any
  convergence-enabling step, the amount of water in cup $\ell + 1$
  must be greater than $E_{\ell}(\ell + 1)$. This is a consequence of
  the fact that the total amount of water in cups
  $1, \ldots, \ell + 1$ is at least $m_\ell = \left(\sum_{i = 1}^{\ell + 1} E_{\ell}(i)\right) - 1$.

    \vspace{.2 cm}

  We now prove a claim analogous to Claim
  \ref{clm:convergence-to-wasted}.
  \begin{claim}
    Suppose $m_\ell \le n \log n$, and consider a segment $S$ in the
    $\ell$-cup phase that focuses on $A_\ell$. If $S$ begins with a
    convergence-enabling step, and every step in $S$ is equilibrium
    converging, then $S$ must contain an emptier-wasted step.
    \label{clm:convergence-to-wasted2}
  \end{claim}
  \begin{proof}
    There are two types of steps in $S$: (1) equilibrium converging
    steps where the emptier uses algorithm $A_\ell$, and (2)
    equilibrium converging steps where the emptier does not use
    $A_\ell$. All type (2) steps are either emptier-wasted or are
    no-ops (meaning that they do not change the state of the cup
    game). On the other hand, because segment $S$ focuses on $A_\ell$,
    there must be at least $2 n \log n$ type (1) steps. Assuming no
    type (2) steps are emptier wasted, then the type (1) steps meet
    the conditions for Claim \ref{clm:convergence-to-wasted} (i.e.,
    the type (1) steps meet the conditions that are placed on $L_i$ in
    the claim). Thus, by Claim \ref{clm:convergence-to-wasted}, at
    least one of the steps in $S$ is emptier-wasted, as desired.
  \end{proof}

  We can now complete the proof of the lemma. For each segment $S$
  that focuses on $A_\ell$, if the number of emptier-wasted steps
  prior to $S$ is less than $\lceil b_\ell \rceil$, then the first
  step of $S$ is convergence-enabling (and any steps in $S$ up until
  the first emptier-wasted step in $S$ are also convergence
  enabling). By the Convergence Property and Claim
  \ref{clm:convergence-to-wasted2}, it follows that $S$ has
  probability at least
  $$p := 1 / \ell^{2|S|} = 1 / \ell^{2n \log n |\mathcal{A}|}$$
  of containing an emptier-wasted step. 

  If $K$ is the number of segments in a phase, then at least
  $K' = K / |\mathcal{A}|$ of the segments in phase $\ell$ must focus
  on algorithm $A_\ell$. Denote these segments by
  $S_1, \ldots, S_{K'}$. Break the $\ell$-cup phase into
  \defn{collections} $C_1, C_2, \ldots, C_{2n \log n}$ of time
  segments, where each $C_i$ contains $K' / (2n \log n)$ of the
  $S_i$'s. Say that a collection $C_i$ \defn{fails} if fewer than
  $\lceil b_\ell \rceil$ emptier-wasted steps occur prior to $C_i$,
  and no emptier-wasted step occurs during $C_i$. Since each $C_i$
  contains at least $K' / (2n \log n)$ segments that focus on
  $A_\ell$, the probability of $C_i$ failing is at most,
  \begin{equation}
    (1 - p)^{K' / (2n \log n)}.
    \label{eq:needKlarge}
  \end{equation}
  Assuming the $K$ is sufficiently large as a function of $n$, then
  the exponent in \eqref{eq:needKlarge} is also sufficiently large as
  a function of $n$, and thus \eqref{eq:needKlarge} is at most
  $1/(4n\log n)$. By a union bound over the $C_i$'s, it follows that
  the probability of any $C_i$ failing is at most $1/2$. On the other
  hand, if none of the $C_i$'s fail, then at least
  $\lceil b_\ell \rceil$ steps must be emptier-wasted (here, we are
  using the fact that the number of collections $C_i$ is
  $2n\log n \ge m_\ell \ge \lceil b_\ell \rceil$). Thus the expected
  number of emptier-wasted steps is at least
  $\lceil b_\ell \rceil / 2$.
\end{proof}

We can now prove Theorem \ref{thm:timestamp}.

\begin{proof}[Proof of Theorem \ref{thm:timestamp}]
  The proof follows exactly as for Theorem
  \ref{thm:general-filler-strategy}, except that Lemma
  \ref{lem:bolus-to-wasted-steps} is replaced with Lemma
  \ref{lem:bolus-to-wasted-steps2}.
\end{proof}

\subsection{Unending guarantees with small resource augmentation}\label{sec:specific_epsilon1}

In this section we show that, even though resource augmentation
$\epsilon > 0$ is needed to achieve unending guarantees for the
smoothed greedy (and asymmetric smoothed greedy) emptying
algorithms, the amount of resource augmentation that is necessary is
substantially smaller than was previously known. In particular, we
prove unending guarantees when $\epsilon = 2^{-\polylog n}$.

Theorem \ref{thm:lowaug} states an unending guarantee for the smoothed
greedy emptying algorithm, using $\epsilon = 2^{-\polylog n}$.

\begin{theorem}
  Consider a single-processor cup game in which the emptier follows
  the smoothed greedy emptying algorithm, and the filler is an
  oblivious filler. If the game has resource augmentation parameter
  $\epsilon \ge 2^{-\polylog n}$, then each step $t$ achieves backlog
  $O(\log \log n)$ with probability $1 - 2^{-\polylog n}$ (where the
  exponent in the polylog is a constant of our choice).
\label{thm:lowaug}
\end{theorem}

Theorem \ref{thm:lowaug2} states an unending guarantee for the
asymmetric smoothed greedy emptying algorithm, using
$\epsilon = 2^{-\polylog n}$.

\begin{theorem}
  Consider a single-processor cup game in which the emptier follows
  the asymmetric smoothed greedy emptying algorithm, and the filler
  is an oblivious filler. If the game has resource augmentation
  parameter $\epsilon \ge 2^{-\polylog n}$, then each step $t$
  achieves tail size $O(\polylog n)$ and the backlog $O(\log \log n)$
  with probability $1 - 2^{-\polylog n}$ (where the exponent in the
  polylog in the probability is a constant of our choice).
\label{thm:lowaug2}
\end{theorem}

We begin by proving Theorem \ref{thm:lowaug}.

Call a step $t$ a \defn{rest step} if the emptier removes less than
$1$ unit of water during that step. The next lemma shows that rest
steps are relatively common.
\begin{lemma}
  Consider a single-processor cup game in which the emptier follows either the smoothed
  greedy emptying algorithm. Any sequence of $n/\epsilon + 1$ steps
  must contain a rest step.
\label{lem:reststep}
\end{lemma}
\begin{proof}
  Whenever a step is a rest step, it must be that every cup contains
  less than $1$ unit of water, meaning that the total amount of water
  in the system is at most $n$. On the other hand, during each
  non-rest step, the amount of water in the system decreases by at
  least $\epsilon$. It follows that, if there are $k$ non-rest steps
  in a row, then the total amount of water in the system after those
  steps is at most $n - k\epsilon$. Thus the number of non-rest steps
  that can occur in a row is never more than $n/\epsilon$, as desired.
\end{proof}

In order to prove Theorem \ref{thm:lowaug}, we will exploit the
following result of Bender et al. \cite{BenderFaKu19} which analyzes
the smoothed greedy algorithm for $\epsilon = 0$:
\begin{theorem}[Bender et al. \cite{BenderFaKu19}]
Consider a single-processor cup game in which the emptier follows the smoothed greedy
emptying algorithm, and the filler is an oblivious filler. Moreover,
suppose $\epsilon = 0$. For any positive constants $c$ and $d$, and any $t \le
2^{\log^c n}$, step $t$ has backlog $O(\log \log n)$ with probability
at least $1 - 2^{-\log^d n}$.
\label{thm:old}
\end{theorem}

Although Theorem \ref{thm:old} only applies to the first
$2^{\polylog n}$ steps of a game, we can use it to prove the following
lemma. Define a \defn{fractional reset} to be what happens if one
reduces the fill $f_j$ of each cup $j$ to $f_j - \lfloor f_j
\rfloor$. That is, the fills of the cups are decreased by integer
amounts to be in $[0, 1)$. The next lemma shows that, if a cup game is
fractionally reset after a given step $j$, then the following steps
$j + 1, j + 2, \ldots, j + 2^{\polylog n}$ are guaranteed to have
small backlog.

\begin{lemma}
  Consider a single-processor cup game in which the emptier follows the smoothed greedy
  emptying algorithm, and the filler is an oblivious filler. Consider
  a step $t_0$, and suppose that, after step $t_0$, the cup system is
  fractionally reset. Then for any positive constants $c$ and $d$, and any $t
  \le t_0 + 2^{\log^c n}$, step $t$ has backlog $O(\log \log n)$ with
  probability at least $1 - 2^{-2\log^d n}$.
  \label{lem:reset}
\end{lemma}
\begin{proof}
  For each cup $j$, let $r_j$ be the random initial offset placed into
  cup $j$ by the smoothed greedy emptying algorithm, and let $c_j$ be
  the total amount of water placed into cup $j$ by the filler during
  steps $1, 2, \ldots, t_0$. Because the emptier always removes water in
  multiples of $1$, they never change the fractional mount of water in
  any cup (i.e., the amount of water modulo $1$). It follows that the
  fractional amount of water in each cup $j$ is given by,
$$q_j := c_j + r_j \pmod 1.$$ Since $r_j$ is uniformly random in $[0,
    1)$, the value $q_j$ is also uniformly random in $[0,
      1)$. Moreover, because the initial offsets $r_j$ are independent
      of one another, so are the values $q_j$.

  Because the values $q_j$ are independent and uniformly random in
  $[0, 1)$, they can be thought of as initial random offsets for the
    smoothed greedy emptying algorithm. Thus, if each cup $j$ is reset
    to have fill $q_j$ after step $t$, then the following steps can be
    analyzed as the first steps of a cup game in which the emptier
    follows the smoothed greedy emptying algorithm. The claimed result
    therefore follows from Theorem \ref{thm:old}.
\end{proof}

We now prove Theorem \ref{thm:lowaug}.
\begin{proof}[Proof of Theorem \ref{thm:lowaug}]
Consider a step $t$, and let $d$ be a large positive constant. For
each step $t_0 \in \{t - n/\epsilon, \ldots, t\}$, Lemma
\ref{lem:reset} tells us that if a fractional reset were to happen
after step $t_0$, then step $t$ would have probability at least $1 -
2^{-\log^d n}$ of having backlog $O(\log \log n)$. By a union bound,
it follows that if a fractional reset were to happen after any of
steps $t - n/\epsilon, \ldots, t$, then step $t$ would have
probability at least $1 - (n/\epsilon)2^{\log^d n}$ of having backlog
$O(\log \log n)$. Supposing that $d$ is a sufficiently large constant,
this probability is at least $1 - 2^{\log^{d/2} n}$.

By Lemma \ref{lem:reststep}, at least one step $t_0 \in \{t -
n/\epsilon, \ldots, t\}$ is a rest step. This means that, at the end
of step $t_0$, every cup contains less than $1$ unit of water. In
other words, the state of the system after step $t_0$ is the
\emph{same} as if the system were to be fractionally reset. It follows
that, for any constant $d$, the backlog after step $t$ is $O(\log \log
n)$ with probability at least $1 - 2^{\log^{d/2} n}$. This completes
the proof.
\end{proof}

The proof of Theorem \ref{thm:lowaug2} follows similarly to the proof
of Theorem \ref{thm:lowaug}. Rather than using Theorem \ref{thm:old},
we instead analyze the case of $\epsilon = 0$ using the following
version of Theorem \ref{thm:main}:
\begin{theorem}
  Consider a single-process cup game that lasts for $2^{\polylog n}$
  steps and in which the emptier follows the asymmetric smoothed
  greedy algorithm. Then with high probability in $n$, the number of
  cups containing $2$ or more or units of water never exceeds
  $O(\polylog n)$ and the backlog never exceeds
  $O(\log \log n)$ during the game.
  \label{thm:maincpy}
\end{theorem}

We now prove Theorem \ref{thm:lowaug2}.
\begin{proof}[Proof of Theorem \ref{thm:lowaug2}]
  The proof follows in exactly the same way as for Theorem
  \ref{thm:lowaug}, except that Theorem \ref{thm:maincpy} is used in
  place of Theorem \ref{thm:old}.
\end{proof}

\subsection{Tight lower bounds on resource augmentation for smoothed greedy}\label{sec:specific_epsilon2}

Theorems \ref{thm:lowaug} and \ref{thm:lowaug2} give unending
guarantees for the smoothed greedy (and asymmetric smoothed greedy)
emptying algorithms using resource augmentation $\epsilon = 1/2^{\polylog n}$. Theorem
\ref{thm:auglb} shows that such guarantees cannot be achieved with
smaller resource augmentation.

\begin{theorem}
  Consider a single-processor cup game on $n$ cups. Suppose
  $\epsilon = 1/2^{\log^{\omega(1)} n}$, and suppose the emptier
  follows either the smoothed greedy emptying algorithm or the
  asymmetric smoothed greedy emptying algorithm. Then there is an
  oblivious filling strategy that causes there to be a step $t$ at
  which the expected backlog is $\omega(\log \log n)$.
  \label{thm:auglb}
\end{theorem}

To prove Theorem \ref{thm:auglb}, we will have the filler follow the
oblivious fuzzing filling algorithm on
$\min(1/\sqrt{\log \epsilon^{-1}}, n)$ cups. Rather than placing water
in multiples of $1/2$, however, the filler now places water in
multiples of $1/2 - \epsilon / 2$ (in order so that the total water
placed in each step is $1 - \epsilon$).

If $\epsilon$ were $0$, then Theorem \ref{thm:general-filler-strategy}
would guarantee an expected backlog of
\begin{equation}
  \Omega\left(\min \left\{\log \frac{1}{\sqrt{\log \epsilon^{-1}}}, \log n \right\} \right) \ge \omega(\log \log n)
  \label{eq:smoothed_backlog}
\end{equation}
after some step $t^* \le 2^{\tilde{O}(1 / \sqrt{\log \epsilon^{-1}})} \le \epsilon^{-1 / 10}$.

The fact that $\epsilon > 0$, however, makes it so that we cannot
directly apply Theorem \ref{thm:general-filler-strategy}. Thus, in
order to prove Theorem \ref{thm:auglb}, we must first prove that the
resource augmentation $\epsilon = 1/2^{\log^{\omega(1)} n}$ is so
small that, with high probability, it does not have a significant
effect on the game by step $t^*$.

In order to bound the impact of resource augmentation on the emptier,
we exploit the random structure of the emptier's algorithm, and use
that random structure to show that the emptier's behavior is robust to
small ``perturbations'' due to resource augmentation.

\begin{lemma}
  Consider resource augmentation $\epsilon = 2^{-2^{\omega(1)} n}$ and
  consider a step $t^* \le \epsilon^{-1 / 10}$. With probability at least
  $1 - O(\sqrt{\epsilon})$, the resource augmentation
  $\epsilon = 1/2^{\log^{\omega(1)} n}$ does not affect the emptier's
  behavior during the first $t^*$ steps.
  \label{lem:augmentation_effect}
\end{lemma}
\begin{proof}
  We begin with a simple observation: the total amount of resource
  augmentation during the first $t^*$ steps is
  $\epsilon t^* \le \epsilon^{9 / 10}$. Call this the
  \defn{Net-Augmentation Observation}.

  For each step $t$, define $S_t$ to be the state of the cup game
  after the $t$-th step without resource augmentation, and define
  $S'_t$ to be the state of the cup game with resource augmentation
  $\epsilon = 1/2^{\log^{\omega(1)} n}$ (note that, in both cases, the
  emptier follows the same variant of the smoothed greedy algorithm
  using the same random initial offsets). Let $S_t(j)$
  (resp. $S'_t(j)$) denote the fill of cup $j$, after the filler has
  placed water in the $j$-th step, but before the emptier has removed
  water (note that, when discussing fill, we include the random
  initial offset placed by the emptier in each cup).

  Suppose that, during steps $1, 2, \ldots, t - 1$, the emptier's
  behavior is unaffected by the resource augmentation. The only way
  that the emptier's behavior filler in step $t$ can be affected by
  the resource augmentation is if either:
  \begin{itemize}
  \item \textbf{Case 1: }$\lfloor S_t(j) \rfloor \neq \lfloor S'_t(j) \rfloor$ for some cup $j$.
    By the Net-Augmentation Observation, it follows that $(S_t(j) \mod 1) \in [-\epsilon^{9/10}, \epsilon^{9/10}]$.
  \item \textbf{Case 2: }$S_t(j) < S_t(k)$ but $S'_t(k) < S'_t(j)$ for
    some cups $j$ and $k$. By the Net-Augmentation Observation, it
    follows that $S_t(k) - S_t(j) < \epsilon^{9/10}$.
  \end{itemize}

  We will show that the probability of either Case 1 or Case 2
  happening is at most $O(\epsilon^{9/10} n^2)$. It follows that the
  probability of resource augmentation affecting the emptier's
  behavior during any of the first $t^*$ steps is at most
  $O(\epsilon^{9/10} n^2 t^*) \le O(\sqrt{\epsilon})$.

  Rather than directly bounding probability of either Case 1 or Case 2
  occurring on step $t$, we can instead bound the probability that either,
  \begin{equation}
   ( S_t(j) \mod 1) \in [-\epsilon^{9/10}, \epsilon^{9/10}]
    \label{eq:cupjbad}
  \end{equation}
  for some cup $j$, or that
  \begin{equation}
    S_t(k) - S_t(j) < \epsilon^{9/10}
    \label{eq:cupsjkbad}
  \end{equation}
  for some cups $j, k$.  Recall that the values
  $S_t(1), S_t(2), \ldots, S_t(n)$ modulo $1$ are uniformly and
  independently random between $0$ and $1$; this is because the
  emptier initially places random offsets $r_j \in [0, 1)$ into each
  cup $j$, which permanently randomizes the fractional amount of water
  in that cup.  Thus the probability that
  $(S_t(j) \mod 1) \in [-\epsilon^{9/10}, \epsilon^{9/10}]$ for a given
  cup $j$ is $O(\epsilon^{9/10})$, and the probability that
  $S_t(k) - S_t(j) < \epsilon^{9/10}$ for a given pair of cups $j, k$
  is also $O(\epsilon^{9/10})$. By union-bounding over all cups $j$
  (for \eqref{eq:cupjbad}) and over all pairs of cups $j, k$ (for
  \eqref{eq:cupsjkbad}), we get that the probability of either
  \eqref{eq:cupjbad} or \eqref{eq:cupsjkbad} occurring is
  $O(\epsilon^{9/10} n^2)$, as desired.
\end{proof}

We can now complete the proof of Theorem \ref{thm:auglb}.
\begin{proof}[Proof of Theorem \ref{thm:auglb}]
  For each step $t$, define $S_t$ and $S'_t$ as in Lemma
  \ref{lem:augmentation_effect}.

  By Theorem \ref{thm:general-filler-strategy} and
  \eqref{eq:smoothed_backlog}, there exists
  $t^* \le 2^{\tilde{O}(1 / \sqrt{\log \epsilon^{-1}})} \le \epsilon^{-1 /
    10}$ for which the expected backlog in $S_{t^*}$ is
  $\omega(\log \log n)$. By the Net-Augmentation Observation (from
  Lemma \ref{lem:augmentation_effect}), it follows that the expected
  backlog in $S'_{t^*}$ is at least
  $\omega(\log \log n - \epsilon^{9/10}) \ge \omega(\log \log n)$,
  which completes the proof.
\end{proof}

\bibliographystyle{abbrv} \bibliography{working}

\end{document}